\documentclass[journal,twoside,web]{ieeecolor}

\usepackage{generic}
\usepackage{cite}
\usepackage{amsmath,amssymb,amsfonts}

\usepackage{algorithmic}
\newtheorem{definition}{Definition}
\newtheorem{proposition}{Proposition}
\newtheorem{theorem}{Theorem}
\newtheorem{remark}{Remark}
\usepackage{graphicx}
\usepackage{hyperref} 
\newtheorem{lemma}{Lemma}
\usepackage{mathtools} 
\let\labelindent\relax
\usepackage{enumitem}
\usepackage{bm}
\usepackage{algorithm,algorithmic}
\hypersetup{hidelinks=true}
\hypersetup{
    hidelinks = false,    
    colorlinks = true,
    linkcolor  = blue,     
    urlcolor   = blue,     
    citecolor  = blue      
}

\usepackage{textcomp}
\def\BibTeX{{\rm B\kern-.05em{\sc i\kern-.025em b}\kern-.08em
    T\kern-.1667em\lower.7ex\hbox{E}\kern-.125emX}}


\begin{document}
\title{A Hamilton-Jacobi Reachability Framework with Soft Constraints for Safety-Critical Systems}
\author{Chams Eddine Mballo, Donggun Lee, and Claire J. Tomlin, \IEEEmembership{Fellow, IEEE}
\thanks{Chams Eddine Mballo and Claire J. Tomlin are with the Department of Electrical Engineering and Computer Sciences, University of California, Berkeley, CA 94704, USA (e-mail: cmballo@berkeley.edu; tomlin@eecs.berkeley.edu).}%
\thanks{Donggun Lee is with the Department of Mechanical and Aerospace Engineering, North Carolina State University, Raleigh, NC 27606, USA (e-mail: dlee48@ncsu.edu).}%
}

\maketitle

\begin{abstract}

Traditional reachability methods provide formal guarantees of safety under bounded disturbances. However, they strictly enforce state constraints as inviolable, which can result in overly conservative or infeasible solutions in complex operational scenarios. Many constraints encountered in practice—such as bounds on battery state-of-charge in electric vehicles, recommended speed envelopes, and comfort constraints in passenger-carrying vehicles—are inherently soft. Soft constraints allow temporary violations within predefined safety margins to accommodate uncertainty and competing operational demands, albeit at a cost such as increased wear or higher operational expenses. This paper introduces a novel soft‑constrained reachability framework that extends Hamilton–Jacobi reachability analysis for the formal verification of safety‑critical systems subject to both hard and soft constraints. Specifically, the framework characterizes a subset of the state space—referred to as the \emph{soft-constrained reach–avoid set}—from which the system is guaranteed to reach a desired set safely, under worst-case disturbances, while ensuring that cumulative soft-constraint violations remain within a user-specified budget. The framework comprises two principal components: (i) an augmented-state model with an auxiliary budget state that tracks soft-constraint violations, and (ii) a regularization-based approximation of the discontinuous Hamilton–Jacobi value function associated with the reach–avoid differential game studied herein. The effectiveness of the proposed framework is demonstrated through numerical examples involving the landing of a simple point-mass model and a fixed-wing aircraft executing an emergency descent, both under wind disturbances. The simulation results validate the framework’s ability to simultaneously manage both hard and soft constraints in safety-critical settings.

\end{abstract}


\section{Introduction}
\label{sec:introduction}
\IEEEPARstart{T}{he} verification and validation of safety-critical systems has been the subject of tremendous research efforts, primarily due to the detrimental consequences of their failure, including loss of human life and significant infrastructure damage.

Hamilton-Jacobi (HJ) reachability is one of several approaches for the verification of complex safety-critical systems. It characterizes a subset of the state-space, known as the \textit{reach-avoid set}, wherein a dynamical system can achieve specific goals, such as reaching a desired set of states while avoiding unsafe states. This subset is implicitly characterized by the viscosity solution of a first-order partial differential equation, known as the HJ Partial Differential Equation (HJ PDE) \cite{c12}. For each state in this subset, an optimal (feedback) strategy can be synthesized from the gradient of this solution, which guarantees goal achievement. Moreover, HJ reachability analysis can accommodate bounded disturbances and model uncertainties via the Hamilton–Jacobi–Isaacs (HJI) PDE \cite{c11}.

HJ reachability analysis has found applications in robotics and aerospace to address performance and safety challenges in autonomous and semi-autonomous systems~\cite{c3, c11}, with practical implementations supported by the Level-Set~\cite{c_11} and HelperOC~\cite{c_12} toolboxes. The method has also been used to develop frameworks for the safe deployment of learning-based control in aerial vehicles~\cite{c14}. Additional applications include ensuring safe locomotion in legged robots~\cite{c15} and providing flight envelope protection for both fixed-wing~\cite{c17} and Electric Vertical Take-off and Landing (EVTOL) aircraft~\cite{c16}.

Prior work on HJ reachability analysis has focused primarily on constraints that must always be satisfied, known as hard constraints. A different category, that of soft constraints, although studied in contexts such as model predictive control (MPC), constrained Markov decision processes (CMDPs), control barrier functions (CBFs), and neural-network--based optimization, has not yet been addressed within the HJ framework. 

Soft constraints in CMDPs, MPC, CBFs, and neural network--based optimization are often introduced to preserve optimization problem feasibility by relaxing hard constraints via slack variables. In contrast, we consider constraints that are inherently soft. Unlike hard constraints, these allow for a certain degree of flexibility, enabling temporary violations that result in increased operational cost rather than immediate system failure. Such constraints are common in a wide range of physical systems. In battery use, for example, manufacturers set operating limits on key variables such as temperature and state-of-charge~\cite{c25} to mitigate degradation mechanisms, including lithium plating, that accelerate battery aging. These limits are soft constraints: they may be exceeded briefly, for instance to meet urgent power demands, but doing so can shorten battery life. Another example from helicopter control involves the engine’s maximum transient torque and temperature limits. These internal constraints fall within the design envelope and are imposed to reduce stress on the engine during flight, thereby extending its life and enhancing overall vehicle safety~\cite{c24}. However, during critical operations—such as a forced landing with conflicting safety priorities—these limits may need to be temporarily exceeded. Soft constraints therefore enable critical trade-offs among safety priorities and between safety and performance.

This paper proposes a novel HJ reachability framework that integrates both hard and soft constraints, offering a unified approach to analyzing and managing safety-critical systems with diverse constraint characteristics.

\subsection{Contributions}
The main contributions are summarized below. 
\begin{enumerate}
    \item This paper introduces a new class of reachability problems that incorporate soft constraints, along with an HJ framework for their characterization and solution. These problems enable the formal treatment of varying safety priorities, ensuring that high priority hard constraints are always satisfied, whereas lower priority soft constraints may be violated in a controlled manner within a predefined, disturbance-robust violation budget.
    
    \item As part of the proposed framework, we introduce an augmentation of the system dynamics with an additional state variable whose evolution enables exact tracking of the duration of soft-constraint violations. While this leads to discontinuous system dynamics, we introduce a novel approach for computing the reachable set of such discontinuous systems by approximating the discontinuous value function with a family of Lipschitz-continuous value functions.

    \item We prove that the framework recovers the classical hard-constrained reachability solution when the violation budget is zero, establishing theoretical consistency and demonstrating that the proposed approach is a strict generalization of existing methods.
    
    \item We demonstrate the practical value of the proposed framework through case studies involving reach-avoid differential games with soft constraints during critical aircraft flight phases. Specifically, we examine two scenarios: the landing of a simple point-mass vehicle model and the emergency descent of a fixed-wing aircraft following propulsion failure, both under wind disturbances. 
\end{enumerate}

\subsection{Related Work}
This section reviews existing primary frameworks used to address problems with reach, avoid, or reach–avoid objectives (i.e., reachability-type problems) under hard and/or soft constraints—namely, CMDPs, MPC and CBFs. Despite substantial progress, these approaches-including classical HJ reachability-still lack a comprehensive and unified methodology for verifying uncertain, safety-critical systems with safety constraints of differing priorities. This gap motivates the contributions of this paper.

\noindent\textbf{CMDPs}~\cite{altman} intrinsically accommodate soft constraints, particularly in safe Reinforcement Learning (RL), by optimizing policies that maximize cumulative reward while ensuring that cumulative constraint costs (i.e., violations) remain below prescribed thresholds~\cite{Russel2020},~\cite{Carrara2018},\cite{Lin2023},\cite{Jiang2024},\cite{Gu2024},\cite{Huang2021},\cite{Fisac2019}. A key strength of CMDPs is the ability to accommodate uncertainty in both the dynamics and the environment. However, standard CMDP formulations do not, in general, ensure persistent satisfaction of state constraints, which can limit their ability to treat hard and soft constraints concurrently~\cite{Yu2022,Ganai2023}. Moreover, CMDP value functions do not, by themselves, provide set-based reachability certificates~\cite{ding2020natural},\cite{ying2023policy},\cite{achiam2017constrained},\cite{chow2018risk}. Thus, an explicit characterization of the set of initial states that ensures task success under both hard and soft constraints is not directly available.

\noindent\textbf{MPC}~\cite{mpc_1}—in both its standard form~\cite{mpc_3},\cite{mpc_4},\cite{mpc_5},\cite{mpc_6},\cite{mpc_7} and its predictive safety–filter variant~\cite{safety_filters}—is widely used to enforce hard constraints (e.g., collision avoidance) in reach–avoid problems. One major advantage of this framework is that, under appropriate problem structure, it admits formulations that reduce to convex programs, for which fast, reliable solvers enable real-time implementation~\cite{mpc_9}. Despite these advantages, ensuring (recursive) feasibility is nontrivial: disturbances and model–plant mismatch can render the optimization infeasible~\cite{mpc_9},\cite{mpc_soft_constraints1},\cite{mpc_soft_constraints2}. A widely adopted remedy is to model soft constraints as relaxations of hard constraints by introducing nonnegative slack variables~\cite{mpc_9},\cite{mpc_soft_constraints1},\cite{mpc_soft_constraints2},\cite{mpc_soft_constraints3},\cite{mpc_soft_constraints4},\cite{mpc_soft_constraints5},\cite{mpc_soft_constraints6},\cite{mpc_soft_constraints7},\cite{mpc_soft_constraints8},\cite{mpc_11}.  

\noindent\textbf{CBFs}~\cite{barrier_functions, barrier_functions2} encode safety constraints as the zero superlevel sets of barrier functions, with the safe set defined as their intersection. A safety filter is derived by solving a constrained optimization problem that enforces constraints to keep the system within the safe set. Its practical merit is most evident for control-affine dynamics, where the safety filter reduces to a convex quadratic program that minimally adjusts the output of a performance-based controller to maintain safety. This quadratic program is typically tractable in real time for moderate-size systems. As with MPC, the optimization problem may be infeasible under conflicting safety constraints. To address this,~\cite{soft_barrier_functions, soft_barrier_functions1} propose a soft-constrained CBF formulation that assigns priorities: high-priority constraints are enforced as hard, whereas lower-priority ones are relaxed via slack variables. A key limitation of CBFs is that the construction of barrier functions is nontrivial—particularly under input constraints—and remains an active area of research.

Soft constraints are also widely used in neural-network–based optimization to enable the application of HJ reachability, CMDPs, MPC, and CBFs to complex, high-dimensional systems and to support the synthesis of safe learning-based controllers. In this setting, hard constraints are typically relaxed into soft constraints by augmenting the training loss with weighted penalty terms~\cite{NN_1},\cite{NN_2},\cite{NN_3}. However, this indirect strategy offers no guarantees at deployment. Thus, recent work has focused on embedding hard constraints directly into the learning process, for example via projection-based methods~\cite{NN_8},\cite{NN_9},\cite{NN_10},\cite{NN_11}.



A salient feature of HJ reachability—unlike CMDPs, MPC, and CBFs—is its ability to provide robust, set-based verification over the state space. Specifically, a single offline computation yields the maximal set of initial states from which reachability properties (e.g., reach--avoid) are guaranteed for nonlinear systems with bounded disturbances and input/state constraints. However, this capability comes at a cost: HJ methods suffer from the curse of dimensionality, with computational complexity that scales exponentially with the state dimension. 

Extending the HJ framework to handle soft constraints preserves rigorous set-based safety guarantees and safety-assured policy synthesis, while enabling explicit trade-offs among safety objectives and between safety and performance. Crucially, our notion of soft constraints departs from prior work: they are not mere slack-based relaxations of hard constraints to avoid infeasibility, but rather lower-priority (yet safety-relevant) requirements whose violation is bounded by a user-chosen, disturbance-robust budget. 
\section{Problem Formulation}
 \label{sec:Problem formulation}
This section introduces a new class of reachability problems that feature both hard and soft constraints. We focus on an instance of this class, the soft-constrained, reach-avoid game.
\subsection{System Dynamics and Game Structure}
\label{sec:system_dynamics}
Consider the ordinary differential equation (ODE)
\begin{equation}
\dot{\mathrm{x}}(s) = f\bigl(s,\,\mathrm{x}(s),\,\bm a(s),\,\bm b(s)\bigr), \quad \mathrm{x}(t) = x.
\label{eq:ode}
\end{equation}
\noindent Here, \(t \in [0,T]\) denotes the initial time, \(s \in [t,T]\) is the time variable, \(\mathrm{x}(\cdot):[t,T] \to \mathbb{R}^n\) is the state trajectory starting at \(\mathrm{x}(t) = x \in \mathbb{R}^n\). The dynamics \(f: [0,T] \times \mathbb{R}^n \times \mathbb{A} \times \mathbb{B} \to \mathbb{R}^n\) satisfies standard assumptions: boundedness, uniform continuity, and Lipschitz continuity with respect to the state~\cite{Evans}. The measurable functions \(\bm{a}(\cdot)\) and \(\bm{b}(\cdot)\), taking values in nonempty, compact sets \(\mathbb{A} \subset \mathbb{R}^{m_a}\) and \(\mathbb{B} \subset \mathbb{R}^{m_b}\), define the control strategies of two players: Player~$\mathrm{A}$, who seeks to steer the system into a target set under constraints, and Player~$\mathrm{B}$, who acts adversarially against Player~$\mathrm{A}$. In practice, Player~$\mathrm{B}$ can be interpreted as a bounded disturbance that opposes the system behavior. 

To capture the worst-case outcome, we model the interaction between the two players as a zero-sum game, where Player~$\mathrm{B}$ responds to Player~$\mathrm{A}$ using a non-anticipative strategy. Let \(\mathcal{A}_t\) and \(\mathcal{B}_t\) denote the sets of measurable functions from \([t,T]\) to \(\mathbb{A}\) and \(\mathbb{B}\), respectively. Player~$\mathrm{A}$ selects a control \(\bm{a} \in \mathcal{A}_t\), and Player~$\mathrm{B}$ responds through a strategy \(\delta \in \Delta_t\), where
\begin{align}
\Delta_t :=\!\{\delta:\mathcal{A}_t\!\to\!\mathcal{B}_t \mid 
&\forall s\!\in\![t,T],~\bm a,\bm{\bar a}\!\in\!\mathcal{A}_t,\;
\bm a\equiv\bm{\bar a}\text{ a.e.\ on }\![t,s] \notag\\[-2pt]
&\Rightarrow\;\delta[\bm a]\equiv\delta[\bm{\bar a}]\text{ a.e.\ on }[t,s]\}.
\label{eq:nonanticipative}
\end{align}
This non-anticipative structure ensures Player~$\mathrm{B}$ cannot use Player~$\mathrm{A}$'s future control inputs to decide its current input.

Under the stated assumptions, each pair \((\bm{a}, \delta[\bm{a}]) \in \mathcal{A}_t \times \mathcal{B}_t\), with \(\delta \in \Delta_t\), induces a unique continuous trajectory \(\phi_{t,x}^{\bm{a}, \delta[\bm{a}]} : [t,T] \to \mathbb{R}^n\) that satisfies the ODE~\eqref{eq:ode} almost everywhere~\cite{c1}.
\subsection{Target and Constraint Sets}\label{sec:constraint_sets}
To capture a wide range of practical scenarios, we consider compact, time-dependent target, and hard and soft constraint sets \(\mathbb{T}_t\), \(\mathbb{C}_{1,t}\), and \(\mathbb{C}_{2,t}\) with uniformly bounded temporal variation\footnote{\(\mathbb{T}_t\) has uniformly bounded temporal variation if there exists \(L>0\) such that
\(d_H(\mathbb{T}_t,\mathbb{T}_s) \le L\,|t-s|\) for all \(s,t\in[0,T]\), where \(d_H\) denotes the Hausdorff distance.}. We represent these time-varying sets on a common space--time domain \([0,T]\times\mathbb{R}^n\) by introducing the lifted sets \(\mathbb{T}=\bigcup_{t\in[0,T]}(\{t\}\times\mathbb{T}_t)\), \(\mathbb{C}_1=\bigcup_{t\in[0,T]}(\{t\}\times\mathbb{C}_{1,t})\), and \(\mathbb{C}_2=\bigcup_{t\in[0,T]}(\{t\}\times\mathbb{C}_{2,t})\), thereby removing explicit time indexing. By Lemma~2 in~\cite{c2} (whose proof appears in~\cite{c3}), each lifted set is a closed subset of \([0,T]\times\mathbb{R}^n\).
Boundedness of each lifted set follows from the fact that \(\mathbb{T}_t\), \(\mathbb{C}_{1,t}\), and \(\mathbb{C}_{2,t}\) are bounded subsets of \(\mathbb{R}^n\) for each \(t \in [0,T]\).
Hence, the lifted sets are compact subsets of \([0,T]\times\mathbb{R}^n\) and can be implicitly represented as the zero sublevel set of a bounded and Lipschitz-continuous function, such as the signed distance function\footnote{Let \(\mathbb{T}\subset[0,T]\times\mathbb{R}^n\) be compact, and let \(\|\cdot\|\) be any norm. The signed distance to \(\mathbb{T}\) is defined by \(g(t,x)=-\inf_{(s,y)\notin \mathbb{T}}\|(t,x)-(s,y)\|\) if \((t,x)\in\mathbb{T}\), and \(g(t,x)=\inf_{(s,y)\in \mathbb{T}}\|(t,x)-(s,y)\|\) if \((t,x)\notin\mathbb{T}\).}; i.e.,
\begin{equation}
\mathbb{T} := \left\{ (t,x) \in[0,T]\times\mathbb{R}^n\Big| g(t,x)\leq 0 \right\},
\label{eq:signed_distance}
\end{equation} where $g$ denotes the signed distance function to $\mathbb{T}$. The sets \(\mathbb{C}_1\) and \(\mathbb{C}_2\) are represented analogously, with signed distance functions \(c_1\) and \(c_2\), respectively.

\subsection{Soft-Constrained Reach-Avoid Problem}
We consider a reach-avoid game in which the soft constraint may be violated for at most \(Q\!\in[0,T]\) time units. This budget \(Q\) defines how long a lower-priority constraint may be relaxed; as such, it enables a principled mechanism for resolving conflicts between competing safety objectives and supports systematic trade-off analysis. For instance, in scenarios where the constraints \(\mathbb{C}_{1,s}\) and \(\mathbb{C}_{2,s}\) are in conflict—e.g., when \(\mathbb{C}_{1,s} \cap \mathbb{C}_{2,s} = \emptyset\) for some \(s \in [0,T]\)—the standard reach-avoid formulation would yield a small or empty reach-avoid set. By distinguishing between hard and soft constraints, the soft-constrained formulation expands the set of verifiable initial conditions to include those from which the reach-avoid task remains achievable with minimal soft constraint violation.

The objective is to characterize the set of initial conditions \((t,x)\) from which the system, against any strategy of Player~\(\mathrm{B}\), is verified to:
(i) reach the target set \(\mathbb{T}_\tau\) at some time \(\tau \in [t,T]\); 
(ii) satisfy the hard constraint \(\mathbb{C}_{1,s}\) for all \(s \in [t,\tau]\); 
(iii) satisfy the soft constraint \(\mathbb{C}_{2,\tau}\) at time \(\tau\) while violating it for at most \(Q\) time units over \([t,\tau)\). Definition~\ref{definition_1} formally defines this set.
\begin{definition}\label{definition_1}
For a violation-time budget \(Q\!\in[0,T]\), the soft-constrained reach–avoid set is defined as   
\begin{align}
\widetilde{\mathcal{RA}}_{Q}&
:= \big\{(t, x) \in [0, T] \times \mathbb{R}^{n} \Big| 
\forall \delta \in \Delta_{t} ,\exists \bm{a} \in \mathcal{A}_{t} , \nonumber\\
&\hspace{3.5em}\exists \tau \in [t, T]:\;
\phi^{\bm{a}, \delta[\bm{a}]}_{t,x}(\tau) \in \mathbb{T}_{\tau} \cap \mathbb{C}_{2,\tau} , \nonumber\\
&\hspace{3.5em}\forall s \in [t, \tau]:\;
\phi^{\bm{a}, \delta[\bm{a}]}_{t,x}(s) \in \mathbb{C}_{1,s} , \nonumber\\
&\hspace{3.0em}\int_{t}^{\tau} \mathbf{1}_{\mathbb{C}_{2,s}^{C}}
\bigl(\phi^{\bm{a}, \delta[\bm{a}]}_{t,x}(s)\bigr) ds \le Q
\big\},
\label{eq:RA_soft}
\end{align}
where \(\mathbf{1}_{\mathbb{C}_{2,s}^{\mathrm{c}}}\!\big(\phi^{\bm a,\delta[\bm a]}_{t,x}(\cdot)\big)\) denotes the indicator function of \(\mathbb{C}_{2,s}^{\mathrm{c}}\) composed with the trajectory \(\phi^{\bm a,\delta[\bm a]}_{t,x}\); it equals \(1\) when \(\phi^{\bm a,\delta[\bm a]}_{t,x}(s)\notin \mathbb{C}_{2,s}\) and \(0\) otherwise.
\end{definition}
\section{Limitation of the Standard HJ Reachability Framework}
\label{sec:original_formulation}
In this section we prove that for \(Q=0\), the soft-constrained reach–avoid problem in Definition~\ref{definition_1} reduces to the classical hard-constrained reach–avoid problem. The reach-avoid set for this classical problem is given by the canonical definition~\cite{c3}: 
\begin{align}
\mathcal{RA}& 
:= \big\{(t,x) \in [0,T] \times \mathbb{R}^{n} \ \Big| \!\!\
\forall \delta \in \Delta_{t},\ \exists \bm{a} \in \mathcal{A}_{t} , \nonumber\\
&\hspace{3.5em}\exists \tau \in [t, T]~\text{ s.t.}~ \phi^{\bm{a},\delta[\bm{a}]}_{t,x}(\tau) \in \mathbb{T}_{\tau} , \nonumber\\
&\hspace{3.5em}\forall s \in [t,\tau]:\; \phi^{\bm{a},\delta[\bm{a}]}_{t,x}(s) \in \mathbb{C}_{1,s} \cap \mathbb{C}_{2,s}
\big\}.
\label{eq:9}
\end{align}

The standard HJ reachability framework is tailored to compute~\eqref{eq:9}. To understand the limitation of this framework and set up a consistent extension, we first recall the standard HJ procedure for computing $\mathcal{RA}$ and formally establish its relationship to $\widetilde{\mathcal{RA}}_{0}$. Specifically, we introduce the classical value function that converts the zero-sum game from a “game of kind” into a “game of degree”.

\begin{proposition}\label{Prop_1}
Let \( V : [0, T] \times \mathbb{R}^n \to \mathbb{R} \) be defined by \begin{align}
V(t,x) = \sup_{\delta \in \Delta_t} \inf_{\bm{a} \in \mathcal{A}_t} \min_{\tau \in [t, T]} \max\Big\{ 
&\max_{s \in [t, \tau]} c_1\big(s, \phi_{t,x}^{\bm{a}, \delta[\bm{a}]}(s)\big), \nonumber\\
\max_{s \in [t, \tau]} c_2\big(s, \phi_{t,x}^{\bm{a}, \delta[\bm{a}]}(s)\big),\;
&g\big(\tau, \phi_{t,x}^{\bm{a}, \delta[\bm{a}]}(\tau)\big) \Big\}.
\label{eq:value_function}
\end{align}
\noindent Then \( V \) is Lipschitz continuous and the reach-avoid set~\( \mathcal{RA}\) coincides with its zero sublevel set: 
\begin{equation}
\mathcal{RA}
= \big\{ (t,x) \in [0,T] \times \mathbb{R}^n \big| \ V(t,x) \le 0 \big\}.
\label{eq:11}
\end{equation}
\end{proposition}
\vspace{0.5em}

This result follows directly from Proposition~2 of~\cite{c3}, upon defining the implicit surface 
for the intersection $\mathbb{C}_{1,t} \cap \mathbb{C}_{2,t}$ via the signed distance function \(c(t,x) := \max\big(c_1(t,x),\,c_2(t,x)\big)\). 

Because \(g\), \(c_1\), and \(c_2\) are signed distance functions as defined in~\eqref{eq:signed_distance}, \(V(t,x)\) captures both target reachability and any constraint violations through pointwise min and max operations evaluated along the trajectory of the system. Specifically, \(V(t,x) \le 0\) certifies that Player~\(\mathrm{A}\) can drive the system from \((t,x)\) to the target while satisfying all constraints, whereas \(V(t,x) > 0\) indicates that Player~\(\mathrm{B}\) can prevent this outcome.
In standard HJ reachability theory, \(V\) can also be characterized analytically as the unique viscosity solution of the following Hamilton–Jacobi variational inequality (HJI-VI).

\begin{proposition}\label{thm1}
The value function $V$ in~\eqref{eq:value_function} is the unique viscosity solution of the HJI variational inequality:
\begin{equation}
\begin{aligned}
0 = \max\Bigl\{
\min\bigl[ g(t,x)-V(t,x),\; \tfrac{\partial V}{\partial t}
           +H\bigl(t,x,\tfrac{\partial V}{\partial x}\bigr) \bigr], \\
\max\bigl(c_1(t,x),\;c_2(t,x)\bigr) - V(t,x)
\Bigr\}, t\!\in\![0,T],~x\!\in\!\mathbb{R}^n.
\label{eq:12}
\end{aligned}
\end{equation}
\noindent where the Hamiltonian and terminal condition are given by
\begin{subequations}
\begin{align}
H\!\bigl(t,x,\tfrac{\partial V}{\partial x}\bigr)
  &= \min_{a\in\mathbb{A}}\max_{b\in\mathbb{B}}
     \frac{\partial V}{\partial x}\cdot f(t,x,a,b), \label{eq:12_a}\\[0.6ex]
V(T,x) &= \max\bigl\{c_1(T,x),\,c_2(T,x),\,g(T,x)\bigr\}.
\label{eq:12_b}
\end{align}
\end{subequations}
\end{proposition}

As in Proposition~\ref{Prop_1}, letting~\(\!c(t,x)\!:=\!\max\big(c_1(t,x),c_2(t,x)\big)\) 
allows Proposition~\ref{thm1} to follow directly from Theorem~1 of~\cite{c3}.

We next formalize the relation between the classical and soft-constrained reach–avoid sets for $\!Q\!=\!0$.
\begin{proposition}\label{prop5_soft}
Let $Q$$=$$0$.~The following equivalence holds:
\begin{equation}
\widetilde{\mathcal{RA}}_{0} = \mathcal{RA}.
\end{equation}
\end{proposition}
The proof follows from the continuity of the trajectory \(s \mapsto \phi^{\bm{\alpha}, \delta[\bm{\alpha}]}_{t,x}(s)\) and the measurability of the map \(s \mapsto \mathbf{1}_{\mathbb{C}_{2,s}^{\mathrm c}}\big(\phi^{\bm{\alpha}, \delta[\bm{\alpha}]}_{t,x}(s)\big)\), and is presented in Appendix~\textcolor{blue}{\hyperref[apx:prop2]{A}}.

Proposition~\ref{prop5_soft} shows that, for $Q=0$, the soft-constrained reach--avoid set exactly recovers the classical set computed by the standard HJ reachability framework. This consistency with the established theory also highlights its limitation: 
for $Q>0$, the standard formulation fails because it enforces all constraints as hard, motivating the extended framework developed in the next section.
\section{A Reachability Framework for Soft-Constrained Reach-Avoid Problems}
This section focuses on the characterization of the soft-constrained reach-avoid set defined in~\eqref{eq:RA_soft} using a new HJ reachability-based framework that generalizes the standard approach introduced in Section~\ref{sec:original_formulation}.

As a first step, we transform the ``game of kind'' in~\eqref{eq:RA_soft} into a ``game of degree'' by introducing an appropriate value function. To facilitate this transformation, we introduce a new set, \(\mathcal{S}\), that augments the state-time space with the budget variable \(Q\). The value function we construct on this augmented space enables the computation of \(\widetilde{\mathcal{RA}}_{Q}\) for all budget levels (i.e., for any \(Q \in [0, T]\)) in a single computation.

Let \(F:[0,T]\to \mathcal{P}\!\big([0,T]\times\mathbb{R}^{n}\big)\)\footnote{%
\(\mathcal{P}([0,T]\times\mathbb{R}^{n})\) denotes the power set of \([0,T]\times\mathbb{R}^{n}\), i.e., the set of all subsets of \([0,T]\times\mathbb{R}^{n}\).} 
be the set-valued map defined by
\begin{equation}\label{eq:Fdef}
F(Q):=\widetilde{\mathcal{RA}}_{Q},
\end{equation}
whose graph is
\begin{equation}\label{eq:gphF}
\mathcal{S}
:= \{ (t,x,Q)\in[0,T]\times\mathbb{R}^{n}\times[0,T] : (t,x)\in F(Q) \}.
\end{equation}
\noindent Next, define the value function\;\(W:\mathcal{S}\!\to\!\mathbb{R}\)\;on the graph\;\(\mathcal{S}\) by
\begin{align}
\hspace*{0.05em} W(t,x,Q) \mathrel{:=}\!\! &\;
\sup_{\delta \in \Delta_{t}} \inf_{a \in \mathcal{A}_{t}} 
\min_{\tau \in [t, T]} \max\Big\{\!
\max_{s \in [t,\tau]} c_1\bigl(s, \phi^{\bm{a}, \delta[\bm{a}]}_{t,x}(s)\bigr), \nonumber\\
& \hspace*{0.5em}
c_2\bigl(\tau, \phi^{\bm{a}, \delta[\bm{a}]}_{t,x}(\tau)\bigr),
\int_t^\tau \!\!\!\mathord{\mathbf{1}}_{\mathbb{C}_{2,s}^{C}}
\bigl(\phi^{\bm{a}, \delta[\bm{a}]}_{t,x}(s)\bigr)\,ds - Q,\;\nonumber\\
& \hspace*{0.5em}g\bigl(\tau, \phi^{\bm{a}, \delta[\bm{a}]}_{t,x}(\tau)\bigr)
\Big\}.
\label{eq:value_function_adjusted}
\end{align}

The value function \(W\), referred to as the soft-constrained value function, retains all the terms present in the classical value function \(V\) from~\eqref{eq:value_function}, except for the term \(\max\limits_{s \in [t, \tau]} c_2\big(s, \phi^{\bm{a}, \delta[\bm{a}]}_{t,x}(s)\big)\). This term is replaced with the expression \(\max\Bigl\{ c_2\big(\tau, \phi^{\bm{a}, \delta[\bm{a}]}_{t,x}(\tau)\big),\ \int_t^\tau\mathbf{1}_{\mathbb{C}_{2,s}^{C}}\big( \phi^{\bm{a}, \delta[\bm{a}]}_{t,x}(s)\big)\, ds - Q \Bigr\}.
\) This modification introduces a relaxation that allows trajectories to violate the soft constraint without incurring penalties for the duration of the violation-time budget \(Q\). The term \(\int_t^\tau \mathbf{1}_{\mathbb{C}_{2,s}^{C}}\big( \phi^{\bm{a}, \delta[\bm{a}]}_{t,x}(s)\big)\, ds - Q\) mathematically formalizes this relaxation. The function \(\mathbf{1}_{\mathbb{C}_{2,s}^{C}}\big(\phi^{\bm{a}, \delta[\bm{a}]}_{t,x}(s)\big)\) is equal to 1 when the trajectory is outside  \(\mathbb{C}_{2,s}\) at time \(s\), and 0 otherwise; therefore, the integral quantifies the total time the trajectory violates the soft constraint. If this quantity is less than or equal to $Q$ (that is, within the allocated budget), then such violations would not cause $W$ to be positive. However, the term \(c_2\big(\tau,\phi^{\bm{a},\delta[\bm{a}]}_{t,x}(\tau)\big)\) ensures that, although violations of the soft constraint are permitted over \([t,\tau)\), the constraint must be satisfied at the terminal time \(\tau\).

The graph \(\mathcal{S}\) admits a geometric characterization as the zero sublevel set of the soft-constrained value function \(W\). Consequently, for any fixed \(Q \in [0,T]\), \(\widetilde{\mathcal{RA}}_{Q}\) is recovered as the zero sublevel set of the function \((t,x) \mapsto W(t,x,Q)\).

\begin{proposition}\label{prop:RA_{Q} characterization} 
\(\mathcal{S} \subseteq [0,T]\times\mathbb{R}^{n}\times[0,T]\) corresponds to the zero sublevel set of the soft-constrained value function \(W\). That is, 
\begin{equation}
\mathcal{S} = \big\{\, (t,x,Q) \in [0,T]\times \mathbb{R}^{n}\times[0,T] \;\big|\; W(t, x, Q) \leq 0 \,\big\}.
\end{equation}
As a result, for each \( Q \in [0, T] \), the soft-constrained reach–avoid set \(\widetilde{\mathcal{RA}}_{Q}\) is given by
\begin{equation}
\hspace{-2.6pt}\widetilde{\mathcal{RA}}_{Q}=\big\{ (t, x)\in[0, T]\times\mathbb{R}^{n}\big| W(t, x, Q)\le0 \big\}.
\label{eq:15}
\end{equation}
\end{proposition}
The proof, detailed in Appendix~\textcolor{blue}{\hyperref[apx:prop4]{B}}, follows from the definitions of \(\mathcal S\) in~\eqref{eq:gphF} and \(W\) in~\eqref{eq:value_function_adjusted}. 
With the characterization of \(\widetilde{\mathcal{RA}}_{Q}\) in \eqref{eq:15}, we proceed to examine its structural properties.

\begin{proposition}\label{prop:monotonicity} Let \(0 \leq Q_1 < Q_2 \leq T\). Then, 
\begin{equation}
\widetilde{\mathcal{RA}}_{Q_1} \subseteq \widetilde{\mathcal{RA}}_{Q_2}.
\end{equation}

\end{proposition}
The proof (see Appendix~\textcolor{blue}{\hyperref[apx:prop5]{C}}) uses the observation that feasibility of the soft-constrained reach–avoid problem in Definition~\ref{definition_1} is preserved under budget enlargement: if \((t,x)\in\widetilde{\mathcal{RA}}_{Q_1}\) for some \(0\le Q_1< Q\), then \((t,x)\in\widetilde{\mathcal{RA}}_{Q}\). Hence the family \(\{\widetilde{\mathcal{RA}}_{Q}\}_{Q\ge 0}\) is monotone nondecreasing in \(Q\). Intuitively, larger \(Q\) yields a larger set of initial conditions from which admissible trajectories can originate, expanding \(\widetilde{\mathcal{RA}}_{Q}\).

Beyond monotonicity, it is also useful to quantify the minimum violation-time budget that, from a given initial condition, enables Player~\(\mathrm{A}\)'s success. The following theorem formalizes this concept through the function \(Q_{\min}\!:\![0, T]\!\times\!\mathbb{R}^n\!\rightarrow\![0, T]\!\cup\!\{+\infty\}\) defined in~\eqref{eq:Q_min_definition}. It also introduces a set-based decomposition of the level sets of \(Q_{\min}\) via \(\widetilde{\mathcal{RA}}_{(t_{1},t_{2}]}\) in~\eqref{eq:Q_min_levelset}. For any \(t_{1}, t_{2} \in [0,T]\) satisfying \(t_{1} < t_{2}\), the set \(\widetilde{\mathcal{RA}}_{(t_{1}, t_{2}]}\) consists of all state--time pairs \((t, x)\) whose minimum required violation-time budget \(Q_{\min}(t, x)\) lies in \((t_{1}, t_{2}]\). This set isolates the state–time pairs \((t,x)\) for which any budget at most \(t_{1}\) is insufficient, which is relevant because some initial conditions may require a strictly positive violation-time budget.
\begin{equation}
Q_{\min}(t, x) \!:= \!
\left\{\!
\begin{array}{ll}
\!+\infty, \text{if } \forall Q \in [0, T],\; W(t, x, Q) > 0, \\[3pt]
\!\min\{ Q \mid W(t, x, Q) \leq 0 \}, \text{otherwise}.
\end{array}
\right.
\label{eq:Q_min_definition}
\end{equation}

\begin{equation}
\widetilde{\mathcal{RA}}_{(t_{1}, t_{2}]} := \{ (t, x) : Q_{\min}(t, x) \in (t_{1}, t_{2}] \}.
\label{eq:Q_min_levelset}
\end{equation}

\begin{theorem}\label{prop6}
The function \(Q_{\min}\) and the set \(\widetilde{\mathcal{RA}}_{(t_1,t_2]}\) satisfy the following for any \(0 \leq t_{1} < t_{2} \leq T\):
\begin{enumerate}[label=\arabic*)]
    \item The function \( Q_{\min} : [0, T] \times \mathbb{R}^n \to [0, T] \cup \{+\infty\} \) is well-defined, in the sense that for every \((t, x) \in [0, T] \times \mathbb{R}^n\) either the set \(\{ Q \in [0, T] : W(t, x, Q) \le 0 \}\) is empty (in which case \(Q_{\min}(t, x) = +\infty\)), or it is nonempty and satisfies \(0 \le Q_{\min}(t, x) \le T\).
    \item The set \( \widetilde{\mathcal{RA}}_{(t_{1}, t_{2}]} \) can be expressed as 
    \[
    \widetilde{\mathcal{RA}}_{(t_{1}, t_{2}]} =
    \widetilde{\mathcal{RA}}_{t_{1}}^{C} \cap \widetilde{\mathcal{RA}}_{t_{2}}.
    \]
    \item As \(t_{1}\) approaches \(t_{2}\) from the left,
    \begin{align}
    \lim_{t_{1}\to t_{2}^-} \widetilde{\mathcal{RA}}_{(t_{1}, t_{2}]} 
    &= \{(t,x): Q_{\min}(t,x) = t_{2}\} \nonumber\\
    &= \widetilde{\mathcal{RA}}_{t_{2}} \cap
      \Bigl(\!\!\bigcap_{\alpha < t_{2}}
      \widetilde{\mathcal{RA}}_{\alpha}^{C} \Bigr).
    \end{align}
\end{enumerate}
\end{theorem}
\begin{proof}
See Appendix~\textcolor{blue}{\hyperref[apx:prop6]{D}}.
\end{proof}

Theorem~\ref{prop6} provides a simple representation of the set \(\widetilde{\mathcal{RA}}_{(t_1, t_2]}\), obtained via a set operation between \(\widetilde{\mathcal{RA}}_{t_{1}}\) and \(\widetilde{\mathcal{RA}}_{t_{2}}\). Furthermore, as \(t_{1} \to t_{2}^{-}\), \(\widetilde{\mathcal{RA}}_{(t_{1}, t_{2}]}\) converges to the level set \(\{ (t, x) : Q_{\min}(t, x) = t_2 \}\), which isolates the initial conditions that require exactly \(t_2\) units of budget. This level set is of practical significance. For example, let \(Q^\star\) denote the lowest violation–time budget, over all initial conditions, required to render the task in Definition~\ref{definition_1} feasible. The “best-case’’ initial conditions are those \((t,x)\) for which \(Q_{\min}(t,x)=Q^\star\). Moreover, for a state \(x_0\) (from such initial conditions), the admissible start times that achieve minimal soft-constraint violation are \(\{\, t\in[0,T]\mid Q_{\min}(t,x_{0})=Q^\star \,\}\); from this set, one may select the largest (resp. smallest) element to delay (resp. expedite) the initiation of the task. More generally, \( \widetilde{\mathcal{RA}}_{(t_{1}, t_{2}]} \) induces a partition of the state-time space into level sets based on \( Q_{\min}(t, x) \), effectively grouping initial conditions by their minimum violation time budget. This partitioning provides a systematic way to quantify and compare the quality of initial conditions. Initial conditions with lower values of \( Q_{\min}(t, x) \) are more favorable, as they require a smaller violation of the soft constraint, whereas those with higher values of \( Q_{\min}(t, x) \) reflect a need for a longer constraint violation.  

To encode soft-constraint violations along system trajectories, the proposed framework augments the dynamics \eqref{eq:ode} with an additional state variable \(z\in\mathbb{R}\) whose evolution measures the cumulative duration of violation. This augmented system enables the computation of \(W\) via dynamic programming and viscosity–solution theory~\cite{c29_1}. Consider the augmented dynamics \(\tilde{f} \!:\! [0,T] \!\times \!\mathbb{R}^n\!\times\mathbb{R} \!\times\!\mathbb{A} \!\times \!\mathbb{B} \!\rightarrow \! \mathbb{R}^n\! \times \!\{-1,0\}\), defined by
\begin{equation}
\begin{bmatrix}\dot{\mathrm{x}}(s)\\ \dot{\mathrm{z}}(s)\end{bmatrix}
\!\!=\!\!\tilde{f}\big(s,\mathrm{x}(s),\mathrm{z}(s),\!\bm a(s),\!\bm b(s)\big)
\!=\!\begin{bmatrix} f\big(s,\mathrm{x}(s),\bm a(s),\bm b(s)\big)\\ -\mathbf{1}_{\mathbb{C}_{2,s}^{\mathrm c}}\!\big(\mathrm{x}(s)\big)\!\end{bmatrix}\!.
\label{eq:augmented_ode}
\end{equation}
\noindent where \(s\in[t,T]\) and the initial condition is \(\begin{bmatrix}\mathrm{x}(t)\\[2pt] \mathrm{z}(t)\end{bmatrix}
= \begin{bmatrix}x\\[2pt] z\end{bmatrix}\).

In \eqref{eq:augmented_ode}, \(\mathrm{x}(\cdot):[t,T]\!\to\!\mathbb{R}^n\)\;denotes the system’s state trajectory. The trajectory \(\mathrm{z}(\cdot):[t,T]\to\mathbb{R}\)\;decreases at unit rate whenever \(\mathrm{x}(s)\notin\mathbb{C}_{2,s}\) and remains constant otherwise. Due to the indicator function, \(\tilde f\) is discontinuous. Consequently, the existence and uniqueness of solutions to \(\tilde{f}\) are not immediate and must be established formally. The following proposition shows that, despite this discontinuity, the system admits a unique absolutely continuous trajectory for any control functions \(\bm{a} \in \mathcal{A}_t\) and \(\bm{b} \in \mathcal{B}_t\).
\begin{proposition}\label{prop:existence_absolute_continuity} For any \(\bm{a} \in \mathcal{A}_t\), \(\bm{b} \in \mathcal{B}_t\), and \((t, x, z) \in [0,T] \times\mathbb{R}^n \times\mathbb{R}\), 
there exists a unique absolutely continuous trajectory 
\(
\tilde{\phi}^{\bm{a},\bm{b}}_{t, x, z}(\cdot): [t,T] \to \mathbb{R}^n \times\mathbb{R}
\) that satisfies the augmented ODE~\eqref{eq:augmented_ode} almost everywhere. Specifically,  
    \begin{align}
        \dot{\tilde{\phi}}^{\bm{a}, \bm{b}}_{t, x, z}(s) &= \tilde{f}(s, \tilde{\phi}^{\bm{a}, \bm{b}}_{t, x, z}(s), \bm{a}(s), \bm{b}(s)), && \text{a.e. } s \in [t, T], \nonumber \\
        \tilde{\phi}^{\bm{a}, \bm{b}}_{t, x, z}(t) &= (x, z).
        \label{eq:augmented_dynamics}
    \end{align}
\end{proposition}
\begin{proof}
See Appendix~\textcolor{blue}{\hyperref[apx:prop7]{E}}.
\end{proof}

Since the first \(n\) components of \(\tilde{f}\) are independent of the \((n+1)^\text{th}\) component, the trajectory \(\tilde{\phi}^{\bm{a},\bm{b}}_{t, x, z}\) can be decomposed as \(\tilde{\phi}^{\bm{a},\bm{b}}_{t, x, z} = (\phi^{\bm{a},\bm{b}}_{t,x};\ \xi^{\bm{a},\bm{b}}_{t, x, z})\), where \(\xi^{\bm{a},\bm{b}}_{t, x, z}\) solves the \((n+1)^\text{th}\) ODE. This decomposition allows rewriting \(W\) as
\begin{align}\label{eq:21_w}
\hspace{-6.0pt}W(t,x,z) \mathrel{:=} &\;\!\!
\sup_{\delta \in \Delta_{t}} \inf_{\bm{a} \in \mathcal{A}_{t}} 
\min_{\tau \in [t, T]} \max\Big\{ 
\max_{s \in [t, \tau]}c_1(s, \phi^{\bm{a}, \delta[\bm{a}]}_{t, x}(s)), \nonumber\\
& c_2(\tau, \phi^{\bm{a}, \delta[\bm{a}]}_{t, x}(\tau)),\;
-\xi^{\bm{a}, \delta[\bm{a}]}_{t, x, z}(\tau),\;
g(\tau, \phi^{\bm{a}, \delta[\bm{a}]}_{t, x}(\tau)) \Big\}.
\end{align}
\noindent In the formulation \eqref{eq:21_w} of \(W\), the initial value \(z\) in \(\xi^{\bm a,\delta[\bm a]}_{t,x,z}(\cdot)\) represents the violation-time budget for the initial condition \((t,x)\) (i.e., \(z=Q\)). The depletion of this budget during the system’s evolution is captured by the trajectory \(\xi^{\bm a,\delta[\bm a]}_{t,x,z}(\cdot)\). Additionally, \(W(t,x,z)\) is well defined for all \((t,x,z)\in[0,T]\times\mathbb{R}^n\times\mathbb{R}\) by Proposition~\ref{prop:existence_absolute_continuity}, which guarantees existence and uniqueness of trajectories under the augmented dynamics.
\subsection{Lipschitz Regularization of the Soft-Constrained Value Function}
Although for each $(t,x,z)$ the trajectory $s\mapsto \tilde{\phi}^{\bm a,\bm b}_{t,x,z}(s)$ is unique and continuous in $s$, the discontinuity of the augmented dynamics makes $W$ discontinuous. Consequently, the standard Hamilton--Jacobi theory~\cite{crandall1983} cannot be applied. From a computational standpoint, the discontinuity of the value function may also lead to numerical instability in computing \(\widetilde{\mathcal{RA}}_{Q}\), particularly when using grid-based solvers that rely on (Lipschitz) regularity assumptions to guarantee convergence and stability~\cite{souganidis_1985}. To address these challenges, this section develops a methodology to characterize the soft-constrained reach–avoid set using a family of Lipschitz-continuous value functions that approximate the discontinuous value function.

We first approximate the indicator function by a family of Lipschitz-continuous functions \( \{h_\epsilon\}_{\epsilon>0} \). For each \( \epsilon>0 \) define
\begin{equation}\label{eq:hepsilon_def}
h_\epsilon(s,x) := \min\bigg(1, \max\Big(0, \frac{1}{\epsilon}\,c_2(s,x)\Big)\bigg).
\end{equation}
\noindent From $\{h_\epsilon\}_{\epsilon>0}$ we define the regularized augmented dynamics $\tilde f_\epsilon$ as
\begin{equation}\label{eq:fep_def}
\tilde f_\epsilon(s,\mathrm{x}(s),\mathrm{z}(s),\bm a(s),\bm b(s)) :=
\begin{bmatrix}
f(s,\mathrm{x}(s),\bm a(s),\bm b(s)) \\
-\,h_\epsilon(s,\mathrm{x}(s))
\end{bmatrix}.
\end{equation}
The time index \(s\) and initial condition are specified in the same manner as in~\eqref{eq:augmented_ode}. The following proposition summarizes key properties of the families \(\{h_\epsilon\}_{\epsilon>0}\) and \(\{\tilde f_\epsilon\}_{\epsilon>0}\).
 
\begin{proposition}\label{introduction_of_lip_func} Let \(\{h_\epsilon\}_{\epsilon>0}\) and \(\{\tilde f_\epsilon\}_{\epsilon>0}\) be the families defined by \eqref{eq:hepsilon_def} and \eqref{eq:fep_def}, respectively. Then the following hold:
\begin{enumerate}[topsep=0pt, itemsep=0pt]
    \item For all \( (s, x) \in [0, T] \times \mathbb{R}^n \),
    \begin{equation}\label{eq:hepsilon_properties}
    \begin{aligned}
    h_\epsilon(s,x) = 0 &\iff c_2(s,x) \leq 0, \\
        0 < h_\epsilon(s,x) < 1 &\iff 0 < c_2(s,x) < \epsilon, \\
    h_\epsilon(s,x) = 1 &\iff \epsilon \leq c_2(s,x).
    \end{aligned}
    \end{equation}
    \item The family \( \{\tilde{f}_\epsilon\}_{\epsilon > 0} \) is monotone decreasing with respect to \( \epsilon \) as \( \epsilon \to 0^+ \), and converges pointwise to \(\tilde{f}\):
    \begin{equation}\label{eq:limit_fe}
    \lim_{\epsilon \to 0^+} \tilde{f}_\epsilon(s,x,z,a,b) = \tilde{f}(s,x,z,a,b),
    \end{equation}
    for all \( (s,x,z,a,b) \in [0, T] \times\mathbb{R}^n \times\mathbb{R}\times \mathbb{A} \times \mathbb{B} \).

    \item For each \( \epsilon > 0 \), the augmented dynamics \( \tilde{f}_\epsilon : [0, T] \times\mathbb{R}^n \times\mathbb{R}\times \mathbb{A} \times \mathbb{B} \rightarrow \mathbb{R}^n \times \{0, 1\} \) is bounded and uniformly continuous. Additionally, \( \tilde{f}_\epsilon \) is Lipschitz continuous in \((x,z)\), uniformly in \( s \),  \( a \) and \( b \).
\end{enumerate}
\end{proposition}
\begin{proof}
See Appendix~\textcolor{blue}{\hyperref[apx:prop8]{F}}.
\end{proof}

Given the family of augmented dynamics \(\{\tilde f_\epsilon\}_{\epsilon>0}\) approximating \(\tilde f\), we construct a family of Lipschitz–continuous soft–constrained value functions \(\{W_\epsilon\}_{\epsilon>0}\) and prove that, as \(\epsilon\to0^+\), this family converges to the discontinuous value function \(W\). Consequently, this yields a regularization-based approximation method for the class of reach–avoid differential games with discontinuous value functions considered in this study. To the best of our knowledge, this is the first such technique in the HJ reachability literature.

\begin{theorem}\label{family_value_functions}
Let \( W_\epsilon(t, x, z) \) denote the value function defined analogously to \( W(t, x, z) \) in~\eqref{eq:21_w}, but associated with the dynamics~\( \tilde{f}_\epsilon \) in~\eqref{eq:fep_def}. Then, the following properties hold:
\begin{enumerate}[topsep=0pt, itemsep=0pt]
    \item \( W_\epsilon(t, x, z) \) is Lipschitz continuous in \( t \), \( x \), and \( z \).
    \item \( \{W_\epsilon\}_{\epsilon > 0} \) is monotonically increasing as \( \epsilon \to 0^+ \).
    \item \( \{W_\epsilon\}_{\epsilon > 0} \) converges pointwise to \( W\) as \( \epsilon \to 0^+ \), i.e.,   
\begin{equation}
\begin{aligned}
\lim_{\epsilon \to 0^+} W_\epsilon(t,x,z) &= W(t,x,z),\\
\text{for all } (t,x,z) &\in [0,T]\times \mathbb{R}^n \times\mathbb{R}.
\end{aligned}
\end{equation}

    \item \( \{W_\epsilon\}_{\epsilon > 0} \) converges to \( W\) in measure, on every compact domain \( D \subset [0, T] \times \mathbb{R}^n\times\mathbb{R}\). That is, for all \( \eta > 0 \),
    \begin{equation}
    \lim_{\epsilon\to0^+} \mu(\{(t, x, z)\in D \big| |W-W_\epsilon|>\eta\}) = 0
    \end{equation} where \( \mu \) is the Lebesgue measure on \( [0, T] \times \mathbb{R}^n \times\mathbb{R}\).
\item For any \((t,x,z)\), the estimation error is bounded as
\begin{align}
\hspace{-1em}\lvert W(t, x, z)-W_\epsilon(t, x, z)\rvert
&\le \mu \Bigl( \bigcup_{\delta\in\Delta_t}
                   \bigcap_{\bm a\in\mathcal A_t}
                   E_{\epsilon,t,x}(\bm a,\delta[\bm a]) \Bigr) \nonumber\\
&\le T - t. \label{eq:mu_bound}\\
&\hspace{-9.8925em}\text{where} \quad \nonumber\\
&\hspace{-9.5625em} E_{\epsilon,t,x}(\bm{a},\delta[\bm{a}])
\!\!= \!\!\{\, s\in[t,T] \!:\! 0 <\! c_{2}(s,\phi^{\bm{a},\delta[\bm{a}]}_{t,x}(s)) \!< \epsilon\}
\label{eq:E_definition}
\end{align}

\end{enumerate}
\end{theorem}
\begin{proof}
See Appendix~\textcolor{blue}{\hyperref[apx:thm2]{G}}.
\end{proof}

Theorem~\ref{family_value_functions} establishes that the family \(\{W_\epsilon\}_{\epsilon>0}\) under-approximates \(W\) and converges to \(W\) pointwise and in measure as \(\epsilon\to0^+\). The proof uses the continuity-from-above property of the Lebesgue measure and Egorov’s theorem~\cite{Royden}. Moreover, Theorem~\ref{family_value_functions} provides an upper bound on the approximation error, namely the Lebesgue measure \(\mu\bigl(\!\bigcup_{\delta\in\Delta_t}\!\bigcap_{\bm a\in\mathcal A_t} E_{\epsilon,t,x}(\bm a,\delta[\bm a])\bigr)\), which is at most \(T-t\) and hence finite. Intuitively, this bound corresponds to the value of a zero-sum game between Players~$\mathrm{A}$ and~$\mathrm{B}$, who seek to minimize and maximize, respectively, the time a trajectory starting at $(t,x)$ spends in an $\epsilon$-neighborhood outside the soft-constraint boundary. As noted in Proposition~\ref{introduction_of_lip_func}, this region is exactly where the indicator function \( \mathbf{1}_{\mathbb{C}_{2,s}^{C}}(x) \) and its approximation \( h_\epsilon(s, x) \) differ. Thus, \( W_\epsilon \) closely approximates \( W \) whenever the measure of this region is near zero.

Since each \( W_\epsilon \) is Lipschitz continuous, we can apply the dynamic programming principle to rigorously derive the HJI equation it satisfies. The Lipschitz regularity of \(W_\epsilon\) also ensures that standard viscosity-solution theory applies.
\begin{proposition}\label{thm2}
Let \(\epsilon\!>0\). For every \(t\in\![0,T)\) and \(\tilde{\delta}\!>0\) with \(t\!+\tilde{\delta}\!\le\!T\), and for all \((x,z) \in \mathbb{R}^n\times\mathbb{R}\), \(W_\epsilon(t,x,z)\) satisfies
\begin{align}
W_\epsilon(t,x,z)
\!&=\! \sup_{\delta\in\Delta_t}
  \! \inf_{a\in\mathcal{A}_t}
   \bigg\{\!
   \min\bigg[
      \!\min_{\tau\in[t,t+\tilde\delta]}
     \!\max\Bigl( \nonumber\\
&\hspace{-0.9cm}   \max_{s\in[t,\tau]}c_1(s,\phi^{\bm{a},\delta[\bm{a}]}_{t,x}(s)),~ 
        c_2(\tau,\phi^{\bm{a},\delta[\bm{a}]}_{t,x}(\tau)),~
        -\xi^{\bm{a},\delta[\bm{a}]}_{t,x,z}(\tau),~\nonumber\\
&\hspace{-0.9cm}
        g(\tau,\phi^{\bm{a},\delta[\bm{a}]}_{t,x}(\tau))
      \Bigr), 
      \max\Bigl(~
              \max_{\tau\in[t,t+\tilde\delta]}
          c_1\!\bigl(\tau,\phi^{\bm{a},\delta[\bm{a}]}_{t,x}(\tau)\bigr), \nonumber\\
&\hspace{-0.9cm}
W_\epsilon\bigl(
            t+\tilde\delta,\,
            \phi^{\bm{a},\delta[\bm{a}]}_{t,x}(t+\tilde\delta),\,
            \xi^{\bm{a},\delta[\bm{a}]}_{t,x,z}(t+\tilde\delta)
        \bigr)
      \Bigr)
   \bigg]
   \bigg\}.
\label{eq:dynamic_prog}
\end{align}
\end{proposition}
\begin{proof}
See Appendix~\textcolor{blue}{\hyperref[apx:thm3]{H}}.
\end{proof}
\noindent Proposition~\ref{thm2} leads to an HJI-VI for which the following theorem establishes that \(W_\epsilon\) is its unique viscosity solution.
\begin{theorem}\label{final_theorem} 
\noindent For any \(\epsilon > 0\), \( W_\epsilon \) is the unique viscosity solution of the following HJI-VI:
\begin{align}\label{visco_solu}
0 &=\!\max\Bigl\{\, c_1(t,x) - W_\epsilon(t,x,z),\;
\min\Bigl[ \max\{c_2(t,x), -z, \nonumber\\
& \hspace*{-0.3em} \quad g(t,x)\} - W_\epsilon(t,x,z),\;
\frac{\partial W_\epsilon}{\partial t}
\!+\!H_{\epsilon}\bigl(t, x, z, \!\frac{\partial W_\epsilon}{\partial x},\! \frac{\partial W_\epsilon}{\partial z}\bigr) \Bigr] \Bigr\}, \nonumber\\
& \hspace*{4.5em} (t,x,z) \in [0,T]\times\mathbb{R}^n\times\mathbb{R}
\end{align}
\noindent where the Hamiltonian and terminal condition are given by
\begin{subequations}\label{eq:visc_solu}
\begin{align}
\hspace{-0.57em}H_{\epsilon}\big(t, x, z, \frac{\partial W_\epsilon}{\partial x},\frac{\partial W_\epsilon}{\partial z}\big)
 &\!=\! \min_{a \in \mathbb{A}} \max_{b \in \mathbb{B}}
\!\left[\!
\begin{pmatrix}
\tfrac{\partial W_{\epsilon}}{\partial x}\\[0.3ex]
\tfrac{\partial W_{\epsilon}}{\partial z}
\end{pmatrix}
\!\!\cdot\!\!
\tilde{f}_{\epsilon}(t,x,z,a,b)\!
\right]\!, \\
&\hspace{-9.5em} W_\epsilon(T, x, z)=\max\{c_1(T, x), c_2(T, x), -z, g(T, x)\}. 
\label{visocity_solu2}
\end{align}
\end{subequations}
\end{theorem}
\begin{proof}
See Appendix~\textcolor{blue}{\hyperref[apx:thm3]{H}}.
\end{proof}

The regularity of the HJI–VI solution \(W_\epsilon\) enables efficient computation with the Level Set~\cite{c_11} and HelperOC~\cite{c_12} toolboxes. Accordingly, we use \(W_\epsilon\) to construct an approximation of the soft-constrained reach–avoid set \(\widetilde{\mathcal{RA}}_{Q}\).

\begin{proposition}\label{Prop_9}
For any \(\epsilon > 0\), the soft-constrained reach-avoid set with \(Q = 0\) can be computed using \(W_{\epsilon}\) as
\begin{equation}\label{eq:RA0_def}
\hspace{-2.55pt}\widetilde{\mathcal{RA}}_{0}
  \!=\!\bigl\{(t,x)\!\in\![0,T]\times\mathbb{R}^n\big|
          W_\epsilon(t,x,0)\!\le\!0\bigr\}.
\end{equation}
\end{proposition}
\begin{proof}
See Appendix~\textcolor{blue}{\hyperref[apx:prop9]{I}}.
\end{proof}

Although \(\widetilde{\mathcal{RA}}_{Q}\) can be exactly recovered from \(W_\epsilon\) when \(Q = 0\), this is not possible for \(Q > 0\). We therefore introduce the notion of the $\epsilon$-approximate soft-constrained reach-avoid set, constructed from $W_\epsilon$, as a proxy for approximating $\widetilde{\mathcal{RA}}_{Q}$. 

For any \(\epsilon>0\) and \(Q\in(0,T]\), the $\epsilon$-approximate soft-constrained reach-avoid set, denoted $\widetilde{\mathcal{RA}}_{Q}^{\epsilon}$, is defined by
\begin{align}
\widetilde{\mathcal{RA}}_{Q}^{\epsilon} 
:=& \Big\{ (t, x) \in [0, T] \times \mathbb{R}^n \;\Big|\; z = Q, \nonumber\\
& \hspace{-6em}W_\epsilon(t, x, z) \leq -\mu\Big( \bigcup_{\delta\in\Delta_t} \bigcap_{\bm{a}\in\mathcal{A}_t} E_{\epsilon, t, x}(\bm{a}, \delta[\bm{a}]) \Big) \Big\}.
\label{eq:soft_reach_avoid_set}
\end{align}
\begin{theorem}\label{Prop_10}
Fix \(Q\in(0,T]\). The family of \(\epsilon\)-approximate soft–constrained reach–avoid sets
\(\{\widetilde{\mathcal{RA}}_{Q}^{\epsilon}\}_{\epsilon>0}\) satisfies:
\begin{enumerate}[label=(\arabic*)]
    \item \(\widetilde{\mathcal{RA}}_{Q}^{\epsilon} \subseteq \widetilde{\mathcal{RA}}_{Q}\) for all \(\epsilon>0\).
    \item As \(\epsilon \to 0^{+}\), \(\widetilde{\mathcal{RA}}_{Q}^{\epsilon}\) converges in measure to \(\widetilde{\mathcal{RA}}_{Q}\); i.e.,
\begin{equation}
\lim_{\epsilon \to 0^{+}}
\mu\bigl(\widetilde{\mathcal{RA}}_{Q}^{\epsilon} \,\Delta\, \widetilde{\mathcal{RA}}_{Q}\bigr)
= 0,
\end{equation}
    where \(\Delta\) denotes the symmetric difference.
\end{enumerate}
\end{theorem}

\begin{proof}
See Appendix~\textcolor{blue}{\hyperref[apx:prop10]{J}}.
\end{proof}

\begin{remark}
The convergence in (2) of Theorem~\ref{Prop_10} holds as stated when the zero-level set of \(W(\cdot,\cdot,Q)\) has measure zero (i.e., is nowhere flat). In general, any residual discrepancy is confined to this zero-level set, and the following holds:
\begin{align}
\lim_{\epsilon\downarrow 0}\mu\bigl(\widetilde{\mathcal{RA}}_{Q}^{\epsilon}\bigr)
&= \mu\bigl(\{(t,x)\in[0,T]\times\mathbb{R}^n : W(t,x,Q)<0\}\bigr) \nonumber\\
&< \mu\bigl(\widetilde{\mathcal{RA}}_{Q}\bigr).
\end{align}
\end{remark}

The proposed method, via the family of Lipschitz–continuous value functions \(\{W_\epsilon\}_{\epsilon>0}\), provides a systematic (\(\epsilon\)-regularized) approximation of the soft-constrained reach–avoid set defined by a discontinuous value function. Proposition~\ref{Prop_9} shows that the approximation is exact when \(Q=0\), whereas for positive budgets (\(Q>0\)) Theorem~\ref{Prop_10} establishes that it is conservative (i.e., safe): every time–state pair \((t,x)\) included in the approximate set satisfies the conditions of Definition~\ref{definition_1}. Moreover, in the positive-budget case, as \(\epsilon\to0^+\) the Lebesgue measure (volume) of the region where the approximate and true sets differ tends to zero.
\section{Numerical Examples}\label{results_section}
The approach-and-landing phase is among the most safety-critical segments of flight, accounting for over 50\% of aircraft accidents~\cite{c2_15}. Motivated by the high risks associated with this phase and the diverse nature of safety constraints in aircraft, we propose computing the reach-avoid set for landing tasks under both hard and soft state constraints. The first example considers a simple one-dimensional point-mass model subject to disturbances. The second example studies a fixed-wing aircraft experiencing a propulsion failure, in the presence of disturbances. Because a go-around is infeasible after a propulsion failure, it is critical to determine, from any state, whether there exists a control strategy that guarantees a safe landing. In this setting, we show that allowing soft-constraint violations substantially enlarges the set of initial conditions (i.e., the reach–avoid set) from which a safe landing can be guaranteed.
\subsection{Safe landing for a one-dimensional point-mass model}
Consider the second-order, time-invariant vertical dynamics
\begin{equation}
    \ddot{y} = u_{1} - g + d_{y},
    \label{eq:14}
\end{equation}
\noindent where \( y \) is the vertical position of the vehicle’s center of mass, \( u_{1}\!\in\! [-60, 60]~\mathrm{m/s}^2 \) is the commanded vertical acceleration due to thrust, \( g\!=\!9.8~\mathrm{m/s}^2 \) is the gravitational acceleration, and \( d_{y}\! \in \![-10, 10]~\mathrm{m/s}^2 \) is a bounded disturbance accounting for unmodeled dynamics and exogenous effects (e.g., wind). Equation~\eqref{eq:14} is a reduced-order model of the 3-DOF planar VTOL (PVTOL) system in~\cite{c2_prime}. Specifically, we neglect roll dynamics by setting the roll angle to zero (\(\phi=0\)) and assuming no coupling between the rolling moment and the lateral acceleration \(\varepsilon=0\). We adopt this simplified model to clearly demonstrate the results in two dimensions.


The task involves landing the one-dimensional point-mass model described by~\eqref{eq:14} from an initial altitude of up to 18 meters, with touchdown occurring within a time horizon of \( T = 1~\mathrm{s} \) and at a vertical speed not exceeding \( 1~\mathrm{m/s} \). The compact target set \(\mathbb{T} := \{(\dot y,y)\in\mathbb{R}^2 \mid \dot y \in [-1,0]~\mathrm{m/s},\; y \in [0,0.7]~\mathrm{m}\}\) encodes this touchdown condition by specifying the permissible vertical position and velocity at landing. To capture operational speed limitations, we assume the vehicle’s design speed range spans from \(-15\) to \(15~\mathrm{m/s}\) for illustrative purposes, representing hard safety limits that must not be violated. To mitigate the effects of significant fatigue loading, which can lead to rapid structural degradation~\cite{c2_15_1}, it is advisable to operate within a safety margin away from the boundaries of the design speed range. Accordingly, a recommended operating range of $-10$ to $10$~m/s is adopted in this study. Temporary excursions beyond the recommended speed limits may be tolerated during extreme events, provided that the design speed limits are not exceeded. These constraints are modeled by the compact sets 
\(\mathbb{C}_1 \) and \( \mathbb{C}_2 \), defined as
\begin{subequations}\label{eq:constraints}
\begin{align}
\mathbb{C}_1 &:= \{(\dot y,y)\in\mathbb{R}^2 \mid \dot y\in[-15,15]~\mathrm{m/s},\; y\in[0,18]~\mathrm{m}\},\\
\mathbb{C}_2 &:= \{(\dot y,y)\in\mathbb{R}^2 \mid \dot y\in[-10,10]~\mathrm{m/s},\; y\in[0,18]~\mathrm{m}\}.
\end{align}
\end{subequations}
\( \mathbb{C}_1 \) encodes the hard constraint on design speed, whereas \( \mathbb{C}_2 \) encodes the soft constraint on recommended speed limits. 

Fig.~\ref{mball1} illustrates the target set \( \mathbb{T}\) (gray with a black boundary), the hard-constraint set \( \mathbb{C}_1 \) (light red), and the soft-constraint set \( \mathbb{C}_2 \) (light brown). The standard reach–avoid set \(\mathcal{RA}\) (see \eqref{eq:9}) is shown with a black dashed boundary. The computed soft-constrained reach–avoid sets for violation-time budgets \(Q\in\{0,\,0.06,\,0.3,\,0.6\}\) are shown with light blue, dark blue, green, and purple boundaries, respectively. We use the HelperOC toolbox~\cite{c_12} to solve~\eqref{eq:12} and~\eqref{visco_solu} for \(t\in[0,T]\) on the compact state domains \(D\) and \(D\times[0,T]\), respectively, where
\(D := \{(\dot y,y)\in\mathbb{R}^2 \mid \dot y\in[-20,20]~\mathrm{m/s},\, y\in[-5,20]~\mathrm{m}\}\), thereby yielding the value functions \(V\) and \(W_\epsilon\). \(V\) is computed on a two-dimensional state-space grid over \((\dot y, y)\); its zero-sublevel set yields the standard reach–avoid set.
Analogously, \(W_\epsilon\) (with \(\epsilon=10^{-3}\)) is computed on a three-dimensional state-space grid over \((\dot y, y, z)\); for a prescribed budget \(Q\), the slice \(z=Q\) provides the corresponding soft-constrained reach–avoid set. All sets in Fig.~\ref{mball1} are reported at the initial time \(t=0\); accordingly, they are shown over the state pair \((\dot y, y)\) only (no explicit time coordinate). To illustrate the practical implications of the computed soft-constrained reach-avoid sets, three representative initial conditions are selected at the boundaries of the sets for $Q = 0$ (light blue dot), $Q = 0.3$ (green dot), and $Q = 0.6$ (purple dot), as shown in Fig.~\ref{mball1}. For each initial condition, the allocated budget equals that of its respective set. The dashed curves, shown in matching colors, depict the resulting trajectories under worst-case disturbances, generated by the state-feedback policy induced by the gradient of \(W_{\epsilon}\).
\begin{figure}[!t]
\centerline{\includegraphics[width=\columnwidth]{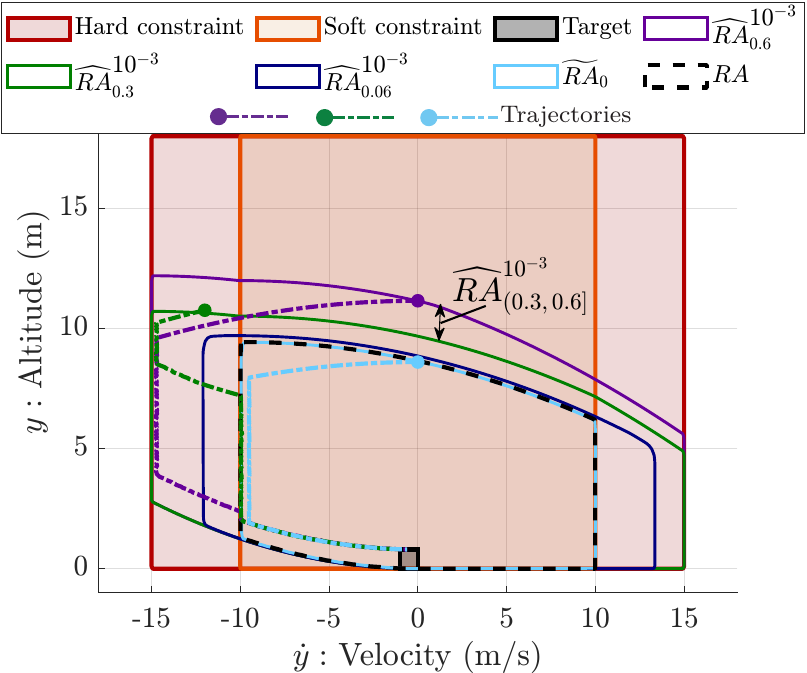}}\caption{Computed soft-constrained reach-avoid sets for $Q \in \{
\textcolor[rgb]{0.4,0.8,1.0}{0},\,
\textcolor[rgb]{0.0,0.0,0.5}{0.06},\,
\textcolor[rgb]{0.0,0.5,0.0}{0.3},\,
\textcolor[rgb]{0.4,0.0,0.6}{0.6}
\}$for the point-mass model. The sets expand as $Q$ increases, with the set at $Q=0$ coinciding with the 
classical reach-avoid set (\(RA\)) computed using the standard HJ framework (Eqs.~\eqref{eq:9} and \eqref{eq:11}). Representative initial conditions 
\textcolor[rgb]{0.4,0.8,1.0}{\(\bullet\)}, 
\textcolor[rgb]{0.0,0.5,0.0}{\(\bullet\)}, 
and \textcolor[rgb]{0.4,0.0,0.6}{\(\bullet\)}
illustrate the validity of the sets: under worst-case disturbances, the feedback controller derived from \(W_{\epsilon}\) generates trajectories (see Fig.~\ref{mball5} for the corresponding budget depletion curves) that satisfy the soft-constrained reach-avoid specifications in Definition~\ref{definition_1}, with budgets not exceeding that of their respective sets. With less budget, the green trajectory leaves the hard constraint boundary earlier than the purple trajectory. Initial conditions whose minimum violation-time budget falls within \((0.3, 0.6]\) are contained in the $Q=0.6$ set, but not in the $Q=0.3$ set.} 
    \label{mball1}
\end{figure}

Consistent with Proposition~\ref{prop5_soft}, Fig.~\ref{mball1} shows that the standard reach--avoid set coincides with the soft-constrained set at zero budget, i.e.,
\(\widetilde{\mathcal{RA}}_{0}=\mathcal{RA}\). To quantify the discrepancy between \( \widetilde{\mathcal{RA}}_{0} \) and \( \mathcal{RA} \), and to assess how grid resolution affects solution accuracy, we compute both sets on progressively finer grids and compare the results. Specifically, \( \mathcal{RA} \) is computed on a 2-D grid of size \( N \! \times \!N \) over \( (\dot{y},y) \), whereas \( \widetilde{\mathcal{RA}}_{0} \) is extracted from the value function \( W_{\epsilon} \), computed on a 3-D grid of size \( N \!\times \!N \times \!200 \) over \( (\dot{y},y,z) \). We vary \( N \) from coarse to fine over \( \{50,\ 100,\ 150,\ 200,\ 300 \} \). For each case, a signed distance function is constructed whose subzero level set corresponds to \( \widetilde{\mathcal{RA}}_{0} \). Then, 50{,}000 points are uniformly sampled from the boundary of \( \mathcal{RA} \), and their distances to the boundary of \( \widetilde{\mathcal{RA}}_{0} \) are evaluated using the constructed function. The mean and maximum absolute signed distances across the sampled points are computed to quantify the discrepancy between \( \widetilde{\mathcal{RA}}_{0} \) and \( \mathcal{RA} \). Fig.~\ref{mball2} shows a logarithmic plot of the mean and maximum boundary errors and the grid spacing \( h = \tfrac{1}{N-1}\sqrt{(\Delta\dot y/40\,\text{m/s})^2 + (\Delta y/25\,\text{m})^2} \) versus the number of grid points $N$, with all quantities expressed in normalized units. With finer grid resolution, both errors decrease. At $N = 300$, they fall below the grid spacing, with mean and maximum errors of 0.0006 and 0.0017, respectively. This resolution is used to compute the sets in Fig.~\ref{mball1} and in all subsequent results for this example.
\begin{figure}[!t]
\centerline{\includegraphics[width=1.0\columnwidth]{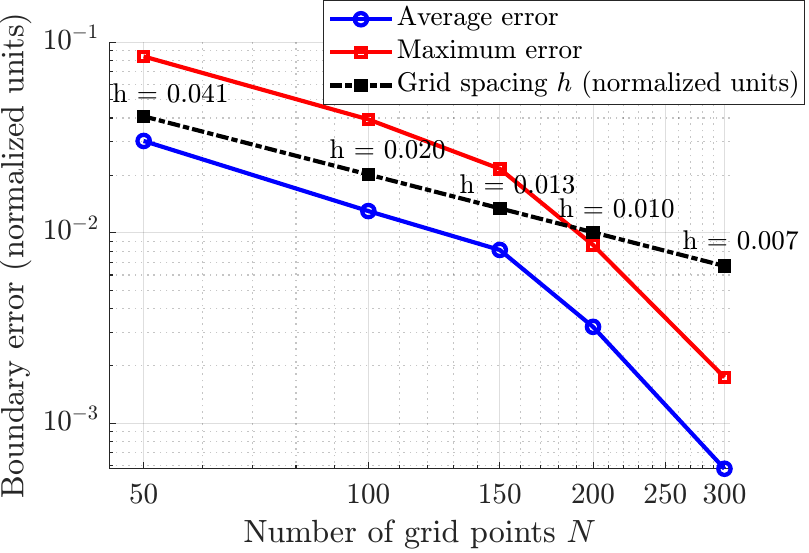}} 
\caption{Boundary error between the soft-constrained ($\widetilde{\mathcal{RA}}_{0}$) and classical (\(RA\),~\eqref{eq:9}) reach-avoid sets, computed using an $N \times N$ discretization of the $(\dot{y},y)$-plane. As $N$ increases, the maximum and average errors drop below the grid spacing (\(h\)), demonstrating convergence.}
\label{mball2}
 \end{figure}

A direct numerical implementation of Theorem~\ref{Prop_10} to obtain an under-approximation of $\widetilde{\mathcal{RA}}_{Q}$ for any nonzero budget would be computationally expensive, as it requires evaluating \( \mu\bigl(\!\mathop{\bigcup}\limits_{\delta \in \Delta_t}\!\mathop{\bigcap}\limits_{\bm{a} \in \mathcal{A}_t}\!E_{\epsilon, t, x}(\bm{a},\, \delta[\bm{a}])\bigr)\) at every grid node on the \( 300 \times 300 \) mesh. To circumvent this challenge, we leverage the convergence-in-measure property, \(\mu(\{|W_\epsilon - W| > \epsilon\}) \to 0\) as \(\epsilon \to 0^{+}\), to construct a computationally tractable proxy, which we use to compute the nonzero-budget sets shown in Fig.~\ref{mball1}. More precisely, on the compact domain \(
\tilde D := [0,T]\!\times\!D\times\![0,T]\), Theorem~\ref{family_value_functions} implies that for every \(\eta>0\) there exists \(\epsilon^\star>0\) such that, for all \(\epsilon \leq \epsilon^*\), the following holds:
\begin{equation}
\mu\bigl(\{(t,(\dot{y}, y), z)\in\tilde D \mid |W_\epsilon-W|>\eta\}\bigr) < \eta.
\end{equation}
Fixing \(\eta = 10^{-3}\), it follows that for all \((t, (\dot{y}, y), z) \in \tilde D\), except on a set of measure less than \(10^{-3}\), 
\begin{equation}
W(t, (\dot{y}, y), z) \leq W_{\epsilon^*}(t, (\dot{y}, y), z) + 10^{-3}.
\label{eq_1}
\end{equation}
Given \(\mu(\tilde D)=1000\), this set occupies less than \(10^{-4}\%\) of the domain volume—a proportion that, for all practical purposes, renders its effect on the computed sets negligible. Hence, \eqref{eq_1} is assumed to hold throughout \(\tilde D\). Consequently, for any \(\epsilon\leq\epsilon^{*}\), we obtain a conservative under-approximation of the true soft-constrained reach–avoid set, \(
\widehat{\mathcal{RA}}_{Q}^{\epsilon}
:= \bigl\{(t,(\dot{y}, y)) \mid W_{\epsilon}(t,(\dot{y}, y),Q) + 10^{-3} \le 0\bigr\}
\subseteq \widetilde{\mathcal{RA}}_{Q}\) on \(D\), except possibly on a set of negligible measure. 

To find \(\epsilon^*\), we consider a strictly decreasing sequence \((\epsilon_k)_{k \in \mathbb{N}_{\geq 0}} \subset \mathbb{R}_{\geq 0}\) and the corresponding sets, \(
\widehat{\mathcal{RA}}_{Q}^{\epsilon_k}.\)
Since the sequence \(\{W_{\epsilon_k}(t,(\dot{y}, y),Q) + 10^{-3}\}_{k \in \mathbb{N}_{\geq 0}}\) forms a Cauchy sequence in measure, the symmetric difference between consecutive sets becomes eventually small. In particular, defining
\begin{equation}
d_{Q,k} := \mu\bigl( \widehat{\mathcal{RA}}_{Q}^{\epsilon_{k}} 
  \,\triangle\, 
  \widehat{\mathcal{RA}}_{Q}^{\epsilon_{k+1}} \bigr),
\end{equation}
we have \(d_{Q,k}\!\le\!2\!\cdot\!10^{-3}\) for all sufficiently large \(k\)\footnote{Let \(A:=\widetilde{\mathcal{RA}}_{Q}\) and 
\(\widehat A_\epsilon:=\widehat{\mathcal{RA}}_{Q}^{\epsilon}\). By the triangle inequality:
\(\mu(\widehat A_{\epsilon_k}\triangle \widehat A_{\epsilon_{k+1}})
\le \mu(\widehat A_{\epsilon_k}\triangle A)+\mu(A\triangle \widehat A_{\epsilon_{k+1}})\).
From \eqref{eq_1}, for all \(\epsilon\le \epsilon^\star\), we have
\(\widehat A_\epsilon \subseteq A\) except on a set of measure \(<10^{-3}\),
hence \(\mu(\widehat A_\epsilon\triangle A)<10^{-3}\).
Therefore, for all sufficiently large \(k\) (so that \(\epsilon_k,\epsilon_{k+1}\le\epsilon^\star\)),
\(
\mu(\widehat A_{\epsilon_k}\triangle \widehat A_{\epsilon_{k+1}})
\le \mu(\widehat A_{\epsilon_k}\triangle A)+\mu(A\triangle \widehat A_{\epsilon_{k+1}})
< 2\cdot 10^{-3}.
\)  
}; that is, there exists \(N\!\in\!\mathbb{N}_{\geq 0}\) such that \(\forall k\!>\!N\), \(\epsilon_k\!<\!\epsilon^*\) and \(d_{Q,k}\!\le\!2\cdot\!10^{-3}\).
Fig.~\ref{mball3&4} shows that for $\epsilon_k \leq 1$, 
the sequence $\bigl(\widehat{\mathcal{RA}}_{Q}^{\epsilon_k}\bigr)_{k \in \mathbb{N}_{\geq 0}}$ converges for $Q \in \{0.06,\,0.3,\,0.6\}$, with no significant change observed as $\epsilon_k$ decreases further. This empirical convergence suggests that \(\epsilon_k\) has fallen below the theoretical threshold \(\epsilon^{*}\) (i.e., \(1<\epsilon^{*}\)), beyond which further reductions in \(\epsilon_k\) yield negligible improvement in the approximation. Accordingly, the soft-constrained reach–avoid sets at budgets \(Q\in\{0.06,0.3,0.6\}\)—shown in Fig.~\ref{mball1} in dark blue, green, and purple, respectively—are computed with \(\epsilon=10^{-3}\), which is well below \(\epsilon^\star\). As \(Q\) increases, the sets expand monotonically, reflecting the greater tolerance for constraint violations.
\begin{figure}[!t]
    \centering
    \includegraphics[width=0.87\columnwidth]{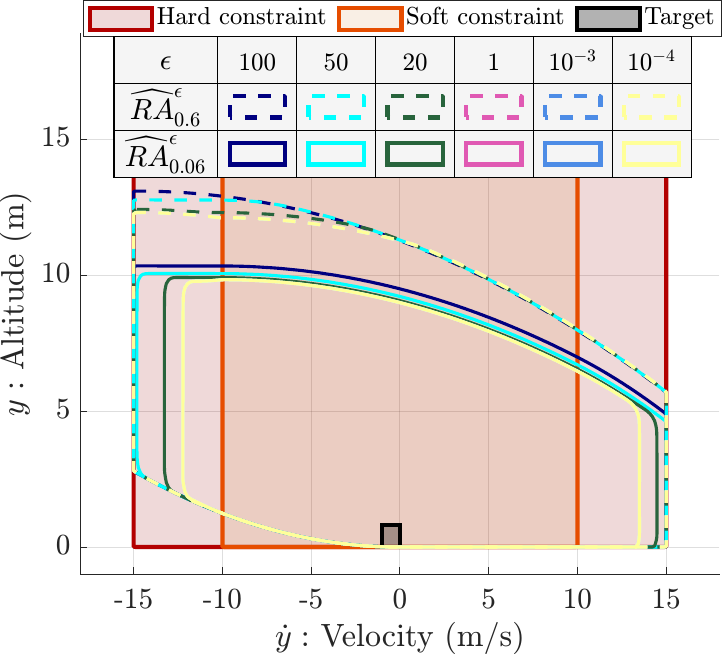}\\[-1ex]
    {\fontsize{8}{8}\selectfont\text{(a)}} \\[1ex]
    
    \includegraphics[width=0.6\columnwidth]{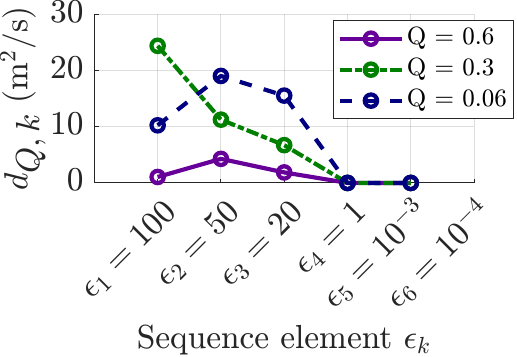}\\[-1ex]
    {\fontsize{8}{8}\selectfont\text{(b)}}
    
    \caption{(a)~Convergence of the approximate soft-constrained reach-avoid set with respect to $\epsilon$. For $\epsilon \leq 1$, the sets become indistinguishable, suggesting that the theoretical threshold has been reached ($1 \leq \epsilon^{*}$). 
    (b)~Convergence in (a) is quantified by the Lebesgue measure of the difference between two consecutive iterates, which is zero for $\epsilon \leq 1$, confirming the conclusion in (a).}
    \label{mball3&4}
\end{figure}

The set \(\widetilde{\mathcal{RA}}_{(0.3, 0.6]}\) consists of initial conditions whose minimum violation-time budget lies in \((0.3, 0.6]\), and is approximated in computation by
\begin{equation}
\widehat{\mathcal{RA}}_{(0.3, 0.6]}^{10^{-3}} 
= \widehat{\mathcal{RA}}_{Q=0.6}^{10^{-3}} 
\cap \bigl( \widehat{\mathcal{RA}}_{Q=0.3}^{10^{-3}} \bigr)^{C}.
\end{equation}
In Fig.~\ref{mball1}, this corresponds to the portion of the set with the purple boundary (\(Q = 0.6\)) that lies outside the set with the green boundary (\(Q = 0.3\)). As the lower bound of the budget interval increases continuously from \(0.3\) toward \(0.6\), this region progressively shrinks until it collapses onto the boundary of the \(Q=0.6\) set. Every state on this curve requires exactly \(0.6\,\mathrm{s}\) of violation time (i.e., \(Q_{\min}(x)=0.6\); see~\eqref{eq:Q_min_definition})
; any smaller budget renders the task infeasible from these states. To make this aspect concrete and illustrate the practical implications of the computed soft-constrained reach-avoid sets, three representative initial conditions are selected at the boundaries of the sets for $Q = 0$, $Q = 0.3$, and $Q = 0.6$, as shown in Fig.~\ref{mball1}. For each initial condition, the allocated budget equals that of its respective set. Under worst-case disturbances, the feedback policy derived from the gradient of $W_{\epsilon}$ drives the system to the target without violating the design-speed limits, while ensuring that violations of the recommended-speed limits remain within the allocated budget. The resulting trajectories and corresponding budget-depletion curves over the landing horizon ($T = 1~\text{s}$) are shown in Figs.~\ref{mball1} and~\ref{mball5}. Due to its smaller budget, the green trajectory departs the hard constraint boundary earlier than the purple trajectory, whereas the light blue trajectory remains entirely within the soft constraint region. In all cases, each trajectory exhausts its allocated budget.
\begin{figure}[!t]
\centerline{\includegraphics[width=0.5\columnwidth]{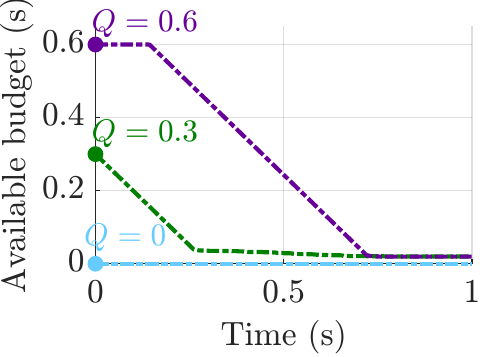}}
\caption{\!\!\!Residual budget along the trajectories of Fig.~\ref{mball1}. Since the initial conditions lie on the boundaries of the sets for $Q \in \{
\textcolor[rgb]{0.4,0.8,1.0}{0},\,
\textcolor[rgb]{0.0,0.5,0.0}{0.3},\,
\textcolor[rgb]{0.4,0.0,0.6}{0.6}
\}$, each trajectory fully expends its allocated budget.}
\label{mball5}
\end{figure}
\subsection{Emergency Landing of a Fixed-Wing Aircraft}

During landing of fixed-wing aircraft the minimum speed, \( V_{\min} \), is set by the stall boundary, whereas the maximum speed, \( V_{\max} \), is determined by structural and aerodynamic design considerations. For practical safety considerations, buffers are maintained above the stall speed and below \( V_{\max} \) — the former recommended by the Federal Aviation Administration~\cite{FAR1990}, and the latter to reduce excessive fatigue loading caused by operating near structural limits. 

Consider the nonlinear three-state longitudinal dynamics of a fixed-wing aircraft under propulsion failure~\cite{c_2_16}:
\begin{equation}
\begin{bmatrix}
    \dot{h} \\
    \dot{V} \\
    \dot{\gamma}
\end{bmatrix}
=
\frac{1}{m}
\begin{bmatrix}
    mV\sin(\gamma) \\
    -D(\alpha, V) - m g \sin(\gamma) +  F_{\text{wind}} \\
    \frac{L(\alpha, V)}{V} - \frac{m g}{V} \cos(\gamma)
\end{bmatrix}.
\label{eq:20}
\end{equation}
\noindent \(h\) denotes the vertical position, \(V\) the airspeed, and \(\gamma\) the flight-path angle. The gravitational acceleration is \( g = 9.8~\mathrm{m/s^2} \), and the lift and drag forces are
\begin{equation}
L(\alpha, V)\!=\!\tfrac{1}{2}\rho S V^2 C_L(\alpha), \quad
D(\alpha, V)\!=\!\tfrac{1}{2}\rho S V^2 C_D(\alpha).
\label{eq:21}
\end{equation}
\noindent The variable \(\alpha \in [0, 13]^\circ\) denotes the angle of attack, \(\rho = 1.225~\mathrm{kg/m^{3}}\) is the air density, and \(S = 112~\mathrm{m^2}\) is the wing area. The functions \(C_{L}(\alpha)\) and \(C_{D}(\alpha)\) are the dimensionless lift and drag coefficients, respectively. The aircraft has a mass of \(m = 60{,}000~\mathrm{kg}\), and the term \(F_{\text{wind}} \in [-10{,}000, 10{,}000]~\mathrm{N}\) represents wind-induced force. The model assumes that the pilot directly controls the angle of attack. All parameters are based on the DC9-30 aircraft~\cite{c17}.

The flight path angle is constrained to \([-3^\circ, 0^\circ]\) to ensure a monotonic descent within a prescribed corridor that avoids overly steep trajectories~\cite{c17}. The stall and maximum descent speeds are \(61~\mathrm{m/s}\) and \(83~\mathrm{m/s}\), respectively~\cite{c17}; applying a \(5~\mathrm{m/s}\) safety margin at both ends yields an operational range of \([66, 78]~\mathrm{m/s}\). The emergency descent begins from an altitude in \([0, 40]~\mathrm{m}\), with touchdown at \(h = 0\). To prevent excessive impact forces that could damage the landing gear, the vertical speed at touchdown must not exceed \(0.91~\mathrm{m/s}\).

Fig.~\ref{mball6&7&8} depicts the target set in gray (Landing zone), the hard constraint set in light red, the soft constraint set in light orange, and the computed soft-constrained reach-avoid sets in light purple (Verified safe region) for $Q = 0$, $5$, and $10$. The terminal time is set to \(T = 10\) seconds, which was determined empirically to be sufficient for the value function to converge to a fixed point during the backward reachability computation.

As \(Q\) increases, the safe region expands, illustrating the additional flexibility afforded by allowing greater violations of the recommended speed limits. Above roughly 22 meters, a safe landing is possible only if the soft constraint is violated, highlighting the critical role of soft constraints in enabling feasible solutions under stringent operational conditions.

To illustrate the computed sets, a representative initial condition was selected from each set, and a trajectory was simulated from it using the system dynamics in~\eqref{eq:20}, under optimal control and worst-case disturbance. Each initial condition was assigned a violation-time budget not exceeding its corresponding set-level budget \(Q\). All resulting trajectories successfully reached the target while respecting the assigned budgets and remaining within the hard constraint bounds. The selected initial conditions, simulated trajectories, and budget usage over time are shown in Figs.~\ref{mball6&7&8} and~\ref{mball9}.
\begin{figure*}[!t]
\centering
\begin{minipage}[t]{0.329\textwidth}
  \centering
  \includegraphics[width=\linewidth]{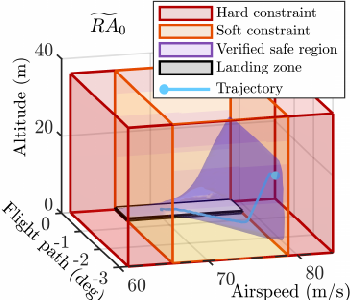}\\[-2pt]
  {\scriptsize (a)}
\end{minipage}\hfill
\begin{minipage}[t]{0.329\textwidth}
  \centering
  \includegraphics[width=\linewidth]{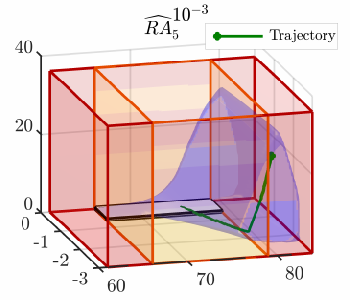}\\[-2pt]
  {\scriptsize (b)}
\end{minipage}\hfill
\begin{minipage}[t]{0.329\textwidth}
  \centering
  \includegraphics[width=\linewidth]{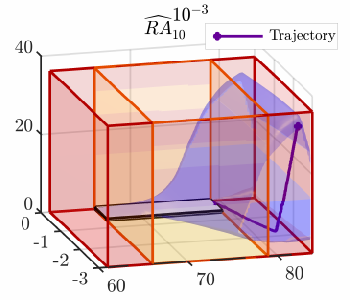}\\[-2pt]
  {\scriptsize (c)}
\end{minipage}

\caption{Computed soft-constrained reach-avoid sets for $Q \in \{\textcolor[rgb]{0.4,0.8,1.0}{0},\,\textcolor[rgb]{0.0,0.5,0.0}{5},\,\textcolor[rgb]{0.4,0.0,0.6}{10}\}$ for the fixed-wing aircraft model under propulsion failure. The sets expand as $Q$ increases, with safe landing above 22 meters only possible via soft-constraint violation. Representative initial conditions 
\textcolor[rgb]{0.4,0.8,1.0}{\(\bullet\)}, 
\textcolor[rgb]{0.0,0.5,0.0}{\(\bullet\)}, 
and \textcolor[rgb]{0.4,0.0,0.6}{\(\bullet\)} are used to validate the sets: under worst-case disturbances, the feedback controller derived from $W_{\epsilon}$ generates trajectories (see Fig.~\ref{mball9} for the corresponding budget depletion curves) that satisfy the soft-constrained reach-avoid specifications in Definition~\ref{definition_1}, with budget usage not exceeding that of their respective sets.}
\label{mball6&7&8}
\end{figure*}

\begin{figure}[!t]
\centerline{\includegraphics[width=0.5\columnwidth]{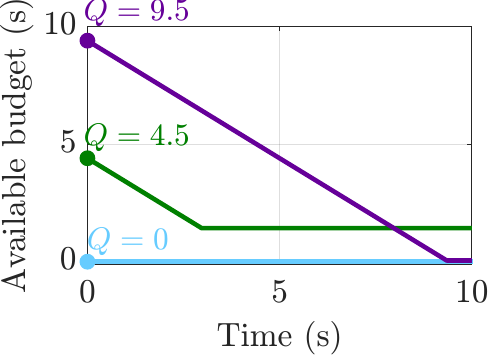}}
\caption{\!\!\!Residual budget along the trajectories of Fig.~\ref{mball6&7&8}. Each trajectory expends at most its allocated budget, which is no greater than the budget of the set containing it.}
\label{mball9}
\end{figure}

\section{Conclusion}
This paper introduces a novel class of reachability problems, termed soft-constrained reachability, which generalizes classical formulations for the verification of safety-critical systems subject to constraints whose violations are tolerable within specified limits. Allowing such violations is not only essential but, under conflicting performance and safety demands, often unavoidable. A representative problem in this class, which is the focus of this study, identifies the set of initial conditions from which the system can robustly reach a target set, satisfy hard constraints at all times, and ensure cumulative soft constraint violations remain within a user-specified budget. A framework based on HJ reachability analysis is developed to compute this set. Importantly, this formulation generalizes the classical reach-avoid problem, recovering it as a special case when the soft constraint budget is set to zero.

By enabling formal safety guarantees under heterogeneous safety constraint types—including both hard and soft constraints—this work addresses a fundamental limitation in classical reachability analysis. As a result, it significantly broadens the applicability of reachability methods to practical, real-world safety-critical systems, where rigid constraint enforcement is frequently infeasible due to conflicting operational demands. Future work will extend this framework to accommodate different types of soft constraints as well as alternative violation metrics. Furthermore, scalable computational methods suitable for high-dimensional systems will be explored.

\appendices
\section*{Appendix} 
\phantomsection 
\label{apx:prop2}
\noindent{\textbf{A.\;Proof of Proposition~\ref{prop5_soft}}}:
\textit{Step 1:} We show that \(\widetilde{\mathcal{RA}}_0 \subseteq \mathcal{RA}\). Let \( Q = 0 \) and assume \((t, x) \in \widetilde{\mathcal{RA}}_{0}\). 
By the definition of \(\widetilde{\mathcal{RA}}_{0}\), \(\forall \delta \in \Delta_t\), \(\exists \bm{\alpha} \in \mathcal{A}_t\), \(\exists \tau_{\bm{\alpha}} \in [t, T]\) such that the trajectory \(\phi^{\bm{\alpha}, \delta[\bm{\alpha}]}_{t,x}(\cdot)\) satisfies the following three conditions simultaneously:  
1) \(\phi^{\bm{\alpha}, \delta[\bm{\alpha}]}_{t,x}(\tau_{\bm{\alpha}}) \in \mathbb{T}_{\tau_{\bm{\alpha}}} \cap \mathbb{C}_{2,\tau_{\bm{\alpha}}}\);  
2) \(\forall s \in [t,\tau_{\bm{\alpha}}],\ \phi^{\bm{\alpha},\delta[\bm{\alpha}]}_{t,x}(s) \in \mathbb{C}_{1,s}\); and  
3) \(\int_t^{\tau_{\bm{\alpha}}} \mathbf{1}_{\mathbb{C}_{2,s}^{C}}(\phi^{\bm{\alpha},\delta[\bm{\alpha}]}_{t,x}(s))\, ds \leq 0\).
Since \(\mathbb{C}_{2,s}^{C}\) is an open set and 
\(\phi^{\bm{\alpha}, \delta[\bm{\alpha}]}_{t,x}(s)\) 
is continuous, the mapping
\(
s \mapsto \mathbf{1}_{\mathbb{C}_{2,s}^{C}}\bigl(\phi^{\bm{\alpha}, \delta[\bm{\alpha}]}_{t,x}(s)\bigr)
\)
is a measurable function. Moreover, because the indicator function is nonnegative, the Lebesgue integral
\(
F(\tau_{\bm{\alpha}}) := \int_t^{\tau_{\bm{\alpha}}} \mathbf{1}_{\mathbb{C}_{2,s}^{C}}\bigl(\phi^{\bm{\alpha}, \delta[\bm{\alpha}]}_{t,x}(s)\bigr) \, ds
\)
is well-defined and nonnegative. Since \( F(\tau_{\bm{\alpha}}) \leq 0 \), it must follow that
\(
F(\tau_{\bm{\alpha}}) = 0.
\) This equality implies that
\(
\mathbf{1}_{\mathbb{C}_{2,s}^{C}}\bigl(\phi^{\bm{\alpha}, \delta[\bm{\alpha}]}_{t,x}(s)\bigr) = 0 \quad \text{for almost every } s \in [t,\tau_{\bm{\alpha}}].
\)
Thus, for almost every \(s \in [t,\tau_{\bm{\alpha}}]\),
\(
\phi^{\bm{\alpha}, \delta[\bm{\alpha}]}_{t,x}(s) \in \mathbb{C}_{2,s}.
\)
This in turn implies that
\(
c_2\bigl(s,\phi^{\bm{\alpha}, \delta[\bm{\alpha}]}_{t,x}(s)\bigr) \le 0 \quad \text{for almost every } s \in [t,\tau_{\bm{\alpha}}].
\)
Since \(c_2(t,x)\) is continuous on \([0,T] \times \mathbb{R}^n\) and \(\phi^{\bm{\alpha}, \delta[\bm{\alpha}]}_{t,x}(s)\) is continuous in \(s\) on \([t,T]\), it follows that \( c_2\bigl(s,\phi^{\bm{\alpha}, \delta[\bm{\alpha}]}_{t,x}(s)\bigr) \le 0~\text{for all}~{s \in [t,\tau_{\bm{\alpha}}]}\). This in turn implies that
\(
\max_{s \in [t,\tau_{\bm{\alpha}}]} c_2\bigl(s,\phi^{\bm{\alpha}, \delta[\bm{\alpha}]}_{t,x}(s)\bigr) \le 0.
\)
\noindent Therefore, it can be concluded that \(\forall \delta \in \Delta_t\), \(\exists \bm{\alpha} \in \mathcal{A}_{t}\), \(\exists \tau_{\bm{\alpha}} \in [t, T]\) such that the trajectory \(\phi^{\bm{\alpha}, \delta[\bm{\alpha}]}_{t,x}(\cdot)\) satisfies the following conditions:  
1) \(\phi^{\bm{\alpha}, \delta[\bm{\alpha}]}_{t,x}(\tau_{\bm{\delta}}) \in \mathbb{T}_{\tau_{\bm{\delta}}}\);  
2) \(\forall s \in [t, \tau_{\bm{\delta}}],~ \phi^{\bm{\alpha}, \delta[\bm{\alpha}]}_{t,x}(s) \in \mathbb{C}_{1,s} \cap \mathbb{C}_{2,s}\).
\noindent Hence, \( (t, x) \in {\mathcal{RA}} \), which implies that 
\(
\widetilde{\mathcal{RA}}_{0} \subseteq \mathcal{RA}
\). \textit{Step 2:} We show that \(\mathcal{RA} \subseteq \widetilde{\mathcal{RA}}_0\). Assume \((t,x) \in \mathcal{RA}\). Then, by definition of \(\mathcal{RA}\), \(V(t,x) \le 0\). Hence, for every \(\delta \in \Delta_t\), there exists \(\bm{\bar{\alpha}} \in \mathcal{A}_t\) and \(\tau_{\bm{\bar{\alpha}}} \in [t,T]\) such that the trajectory \(\phi^{\bm{\bar{\alpha}},\delta[\bm{\bar{\alpha}}]}_{t,x}(\cdot)\) satisfies all the following conditions:  
1) \(\phi^{\bm{\bar{\alpha}},\delta[\bm{\bar{\alpha}}]}_{t,x}(\tau_{\bm{\bar{\alpha}}}) \in \mathbb{T}_{\tau_{\bm{\bar{\alpha}}}}\);  
2) \(\forall\, s \in [t,\tau_{\bm{\bar{\alpha}}}],~ \phi^{\bm{\bar{\alpha}},\delta[\bm{\bar{\alpha}}]}_{t,x}(s) \in \mathbb{C}_{1,s} \cap \mathbb{C}_{2,s}\).
Since \(\phi^{\bm{\bar{\alpha}},\delta[\bm{\bar{\alpha}}]}_{t,x}(s) \in \mathbb{C}_{2,s}\) for all \(s \in [t,\tau_{\bm{\bar{\alpha}}}]\), the indicator function 
\(
\mathbf{1}_{\mathbb{C}_{2,s}^{C}}\bigl(\phi^{\bm{\bar{\alpha}},\delta[\bm{\bar{\alpha}}]}_{t,x}(s)\bigr)
\)
is zero for all \(s \in [t,\tau_{\bm{\bar{\alpha}}}]\). Consequently,
\(
\int_t^{\tau_{\bm{\bar{\alpha}}}} \mathbf{1}_{\mathbb{C}_{2,s}^{C}}\bigl(\phi^{\bm{\bar{\alpha}},\delta[\bm{\bar{\alpha}}]}_{t,x}(s)\bigr)\,ds = 0.
\)
Thus, \(\phi^{\bm{\bar{\alpha}},\delta[\bm{\bar{\alpha}}]}_{t,x}(\cdot)\) satisfies the following conditions:  
1) \(\phi^{\bm{\bar{\alpha}},\delta[\bm{\bar{\alpha}}]}_{t,x}(\tau_{\bm{\bar{\alpha}}}) \in \mathbb{T}_{\tau_{\bm{\bar{\alpha}}}} \cap \mathbb{C}_{2,\tau_{\bm{\bar{\alpha}}}}\);  
2) \(\forall\, s \in [t,\tau_{\bm{\bar{\alpha}}}],~ \phi^{\bm{\bar{\alpha}},\delta[\bm{\bar{\alpha}}]}_{t,x}(s) \in \mathbb{C}_{1,s}\);  
3) \(\int_t^{\tau_{\bm{\bar{\alpha}}}} \mathbf{1}_{\mathbb{C}_{2,s}^{C}}\bigl(\phi^{\bm{\bar{\alpha}},\delta[\bm{\bar{\alpha}}]}_{t,x}(s)\bigr)\,ds \le 0\).  
Hence, \((t,x) \in \widetilde{\mathcal{RA}}_0\), which implies \(\widetilde{\mathcal{RA}}_0 \supseteq \mathcal{RA}\).
\noindent Combining both cases shows that \( \mathcal{RA} = \widetilde{\mathcal{RA}}_{0} \). \hfill$\blacksquare$\
\vspace{0.2em}
\phantomsection           
\label{apx:prop4}               
\noindent{\textbf{B.\;Proof of Proposition~\ref{prop:RA_{Q} characterization}}}:
We prove that:(i) \((t,x,Q)\!\in\!\mathcal{S}\!\Rightarrow\! W(t,x,Q)\le 0\); 
(ii) \(W(t,x,Q)\!\le\!0\!\Rightarrow\!(t,x,Q)\in\mathcal{S}\).\\
\textit{(i)} Suppose, for contradiction, that \((t, x, Q)\!\in\!\mathcal{S}\) and \(W(t,x,Q)>0\). Define the function \(C\) as
\begingroup            
\setlength{\abovedisplayskip}{0.3ex}     
\setlength{\abovedisplayshortskip}{0pt}
\setlength{\belowdisplayskip}{0.3ex}   
\setlength{\belowdisplayshortskip}{0pt}
\begin{equation*}
\begin{aligned}[t]
& \hspace{0.1em} C(t,x,Q,\tau,\delta,\bm{a}) \mathrel{:=} \max\Big\{\max_{s\in[t,\tau]} c_1(s,\phi^{\bm{a},\delta[\bm{a}]}_{t,x}(s)), \\ 
&\hspace{0.1em} c_2(\tau,\phi^{\bm{a},\delta[\bm{a}]}_{t,x}(\tau)), \int_t^\tau\!\!\displaystyle \!\!\mathord{\mathbf{1}}_{\mathbb{C}_{2,s}^{C}}(\phi^{\bm{a},\delta[\bm{a}]}_{t,x}(s))\,ds\!-\!Q, g(\tau, \phi^{\bm{a},\delta[\bm{a}]}_{t,x}(\tau))\!\Big\}
\end{aligned}
\end{equation*}
\endgroup
Therefore, \(W\) can be re-written as \\
\(\begin{aligned}
W(t,x,Q) &= \sup_{\delta \in \Delta_t} \inf_{\bm{a} \in \mathcal{A}_t} \min_{\tau \in [t,T]} C(t,x,Q,\tau,\delta,\bm{a}) \\[0.1ex]
\end{aligned}\). \\
Since \(W(t,x,Q)>0\), there exist \(\epsilon>0\) and \(\bar\delta\in\Delta_t\) such that \(C(t,x,Q,\tau,\bar\delta,\bm a)>\epsilon\) for all \(\bm a\in\mathcal A_t\) and all \(\tau\in[t,T]\).\\
\noindent On the other hand, because \((t,x,Q)\in\mathcal{S}\), for every
\(\delta\in\Delta_t\) there exist
\(\bm a_\delta\in\mathcal{A}_t\) and \(\tau_\delta\in[t,T]\) such that  
(1) \(\phi^{\bm a_\delta,\delta[\bm a_\delta]}_{t,x}(\tau_\delta)\in
     \mathbb{T}_{\tau_\delta}\cap\mathbb{C}_{2,\tau_\delta}\),  
(2) \(\displaystyle\int_{t}^{\tau_\delta}
     \mathbf{1}_{\mathbb{C}_{2,s}^{C}}\!\bigl(
     \phi^{\bm a_\delta,\delta[\bm a_\delta]}_{t,x}(s)\bigr)\,ds\le Q\), and  
(3) \(\phi^{\bm a_\delta,\delta[\bm a_\delta]}_{t,x}(s)\in
       \mathbb{C}_{1,s}\) for all \(s\in[t,\tau_\delta]\).
Therefore, with \(\delta = \bar{\delta}\) and using~\eqref{eq:signed_distance}, the following hold: \(
\begin{aligned}
&\text{(a)}~\!\!\max_{s \in [t,\tau_{\bar{\delta}}]} c_1\bigl(s, \!\phi^{\bm{a}_{\bar{\delta}}, \bar{\delta}[\bm{a}_{\bar{\delta}}]}_{t,x}(s)\bigr)\!\leq\!0, 
\!\!\! &\text{(b)}~\!c_2\bigl(\tau_{\bar{\delta}}, \phi^{\bm{a}_{\bar{\delta}}, \bar{\delta}[\bm{a}_{\bar{\delta}}]}_{t,x}(\tau_{\bar{\delta}})\bigr)\!\leq 0, 
\end{aligned}
\)
\(
\begin{aligned}
&\text{(c)}~\int_t^{\tau_{\bar{\delta}}} \mathord{\mathbf{1}}_{\mathbb{C}_{2,s}^{C}}\bigl(\phi^{\bm{a}_{\bar{\delta}}, \bar{\delta}[\bm{a}_{\bar{\delta}}]}_{t,x}(s)\bigr)\,ds - Q \leq 0, \\
&\text{(d)}~g\bigl(\tau_{\bar{\delta}}, \phi^{\bm{a}_{\bar{\delta}}, \bar{\delta}[\bm{a}_{\bar{\delta}}]}_{t,x}(\tau_{\bar{\delta}})\bigr) \leq 0.
\end{aligned}
\)\\
Therefore, \(C(t, x, Q,\tau_{\bar{\delta}}, \bar{\delta}, \bm{a}_{\bar{\delta}})\leq0\), contradicting the earlier conclusion that \(C(t, x, Q, \tau, \bar{\delta}, \bm{a})> \epsilon>0\) for all \(\bm{a} \in \mathcal{A}_t\) and \(\tau \in [t, T]\). Hence, the assumption \(W(t, x, Q)>0\) leads to a contradiction; this implies\((t,x,Q)\in  \mathcal{S} \Longrightarrow W(t, x, Q) \leq 0.\)\\
\textit{(ii)} Let \((t,x,Q)\)  be such that \(W(t,x,Q)\!\leq\!0\) and suppose, for contradiction, that \( (t,x,Q)\!\notin\!\mathcal{S}\). Then, by definition, \(\exists\tilde{\delta}\!\in\!\Delta_t\) such that \(\forall \bm{a} \in \mathcal{A}_t,~\forall \tau \in [t,T]\), at least one of the following holds: 
\((1)\ \phi_{t,x}^{\bm{a},\tilde{\delta}[\bm{a}]}(\tau) \notin \mathbb{T}_\tau \cap \mathbb{C}_{2,\tau}\); 
\((2)\ \exists s \in [t,\tau]\) such that \(\phi_{t,x}^{\bm{a},\tilde{\delta}[\bm{a}]}(s) \notin \mathbb{C}_{1,s}\); 
\((3)\ \int_t^\tau \mathbf{1}_{\mathbb{C}_{2,s}^C}\bigl(\phi_{t,x}^{\bm{a},\tilde{\delta}[\bm{a}]}(s)\bigr)\,ds\,>\,Q\). Hence, \(\exists\,\eta\!>\!0\) such that \(\forall \bm{a} \in \mathcal{A}_t,~\forall \tau \in [t,T]\), 
\(C(t,x,Q,\tau,\tilde{\delta},\bm{a})\!>\!\eta.
\) Since \(W(t,x,Q)\le 0\) and
\( \begin{aligned}
W(t,x,Q) &= \sup_{\delta \in \Delta_t} \inf_{\bm{a} \in \mathcal{A}_t} \min_{\tau \in [t,T]} C(t,x,Q,\tau,\delta,\bm{a}) \leq 0, 
\end{aligned}\)
it follows that for every \(\epsilon>0\) and every \(\delta\in\Delta_t\), there exist
\(\tilde{\bm a}\in\mathcal A_t\) and \(\tilde{\tau}\in[t,T]\) such that
\(
C\bigl(t,x,Q,\tilde{\tau},\delta,\tilde{\bm a}\bigr)\le \epsilon .
\)
\noindent Setting \( \epsilon = \eta/2 \) and \( \delta = \tilde{\delta} \) leads to a contradiction. Hence, the assumption \( (t,x,Q) \notin \mathcal{S}\) is false, and we conclude \( W(t,x,Q) \leq 0 \;\Rightarrow\; (t,x,Q)\in\mathcal{S}\).\hfill$\blacksquare$\ \\

\vspace{-0.95em}
\phantomsection 
\label{apx:prop5}           
\!\!\!\!\!\!\noindent{\textbf{C.\;Proof of Proposition~\ref{prop:monotonicity}}}:
Let \(0 \leq Q_1 < Q_2 \leq T\) and \( (t, x) \in \widetilde{\mathcal{RA}}_{Q_1}\). Let \(C\) denote the function introduced in the proof of Proposition~\ref{prop:RA_{Q} characterization}. For all \(\bm{a} \in \mathcal{A}_t\), \(\delta \in \Delta_t\), and \(\tau \in [t,T]\) \\
\(\begin{aligned}
  & C(t,x,Q_{2},\tau,\delta,\bm{a}) \leq C(t,x,Q_{1},\tau,\delta,\bm{a})\\
\implies\quad & W(t,x,Q_{2}):=\sup_{\delta \in \Delta_t} \inf_{\bm{a} \in \mathcal{A}_t} \min_{\tau \in [t,T]} C(t,x,Q_{2},\tau,\delta,\bm{a}) \leq \\ 
& \sup_{\delta \in \Delta_t} \inf_{\bm{a} \in \mathcal{A}_t} \min_{\tau \in [t,T]} C(t,x,Q_{1},\tau,\delta,\bm{a}):=W(t,x,Q_{1}).
\end{aligned}
\)\\
\noindent\(\forall (t, x)\!\in\!\widetilde{\mathcal{RA}}_{Q_1}\), Proposition~\ref{prop:RA_{Q} characterization} gives \(W(t,x,Q_1)\!\leq\!0\). Thus, \(W(t,x,Q_2)\!\leq\!0\), so \((t, x)\!\in\! \widetilde{\mathcal{RA}}_{Q_2}\). Hence, \(\widetilde{\mathcal{RA}}_{Q_1}\!\subseteq\!\widetilde{\mathcal{RA}}_{Q_2}\).\hfill$\blacksquare$\ \\
\\
\phantomsection 
\label{apx:prop6}              
\!\!\noindent{\textbf{D.\;Proof of Theorem~\ref{prop6}}}: \textbf{1)} Since \(W(t, x, Q)\) is continuous in \(Q\), the set \(S_{t,x} := \{Q \in [0, T] : W(t, x, Q) \leq 0\}\) is compact as it is bounded (it is contained in \([0, T]\)) and closed (it is the preimage of the closed set \((-\infty, 0]\) under \(W(t,x,\cdot)\)). By the extreme value theorem, if \(S_{t,x}\!\neq\!\emptyset\), then \(W(t,x,\cdot)\) attains a unique minimum on \(S_{t,x}\), and we define \(Q_{\min}(t,x)\!:=\!\min S_{t,x}\!\in\![0,T]\). If \(S_{t,x}\!=\!\emptyset\), set \(Q_{\min}(t,x)\!:=\!+\infty\). Thus, \(Q_{\min}(t,x)\) is well-defined \(\forall (t,x) \in [0,T] \times \mathbb{R}^n\).\\
\noindent \textbf{2)} First, suppose \((t,x)\! \in \! \widetilde{\mathcal{RA}}_{(t_{1},t_{2}]}\). Then there exists \(Q \!\in \! (t_{1},t_{2}]\) such that \(W(t,x,Q) \le 0\), and \(W(t,x,\tilde{Q})\! > \!0\) for all \(\tilde{Q}\!\in\![0,t_{1}]\). By Propositions~\ref{prop:RA_{Q} characterization} and~\ref{prop:monotonicity}, this implies \((t,x)\!\in\!\widetilde{\mathcal{RA}}_{t_{2}}\) and \((t,x)\!\notin\!\widetilde{\mathcal{RA}}_{t_{1}}\), so \(\widetilde{\mathcal{RA}}_{(t_{1},t_{2}]} \!\subseteq \!\widetilde{\mathcal{RA}}_{t_{1}}^C \cap \widetilde{\mathcal{RA}}_{t_{2}}\). Conversely, assume \((t,x)\!\in\!\widetilde{\mathcal{RA}}_{t_{1}}^C\!\cap\! \widetilde{\mathcal{RA}}_{t_{2}}\). Then \(W(t,x,Q) > 0\) for all \(Q \in [0,t_{1}]\), and there exists \(\tilde{Q} \in [0,t_{2}]\) with \(W(t,x,\tilde{Q}) \le 0\). Thus, \(Q_{\min}(t,x) \in (t_{1},t_{2}]\), so \((t,x) \in \widetilde{\mathcal{RA}}_{(t_{1},t_{2}]}\), and hence \(\widetilde{\mathcal{RA}}_{(t_{1},t_{2}]} \supseteq \widetilde{\mathcal{RA}}_{t_{1}}^C \cap \widetilde{\mathcal{RA}}_{t_{2}}\). Combining both inclusions yields \(\widetilde{\mathcal{RA}}_{(t_{1},t_{2}]} = \widetilde{\mathcal{RA}}_{t_{1}}^C \cap \widetilde{\mathcal{RA}}_{t_{2}}\).\\
\noindent \textbf{3)}~For \(t_1 < t_2 \in [0, T]\), let \(\{t_n\} \subset (t_1, t_2)\) be an increasing sequence with \(t_n \to t_2\) as \(n \to \infty\). Then,\\
\(
\begin{aligned}
\lim_{n \to \infty} \widetilde{\mathcal{RA}}_{(t_{n}, t_{2}]} 
&= \lim_{n \to \infty} \{ (t, x) : Q_{\min}(t, x) \in (t_{n}, t_{2}] \} \\
\begin{array}{c}\scriptstyle (\text{Since } (t_{n}, t_{2}] \supseteq \\ \scriptstyle (t_{n+1}, t_{2}]\ \forall n)\end{array} 
&= \bigcap_{n=1}^\infty \{ (t, x) : Q_{\min}(t, x) \in (t_{n}, t_{2}] \} \\
&= \{ (t, x) : Q_{\min}(t, x) = t_{2} \}.
\end{aligned}
\)\\
Since the limit along any increasing sequence in \((t_1,t_{2})\) converging to \(t_{2}\) is identical, it follows that
\(
\lim_{t_1\to t_2^-} \widetilde{\mathcal{RA}}_{(t_1,t_2]}
=\{(t,x): Q_{\min}(t,x)=t_2\}.
\)
As shown earlier
\(
\widetilde{\mathcal{RA}}_{(t_{1},t_{2}]} = \widetilde{\mathcal{RA}}_{Q=t_{1}}^C \cap \widetilde{\mathcal{RA}}_{Q=t_{2}}.
\)
Thus, taking the limit as \(t_{1}\to t_{2}^-\) yields\\
\(
\begin{array}{rcl}
\lim\limits_{t_1\to t_2^-} \widetilde{\mathcal{RA}}_{(t_{1}, t_{2}]} 
& = & \lim\limits_{t_1\to t_2^-}
\widetilde{\mathcal{RA}}_{t_{1}}^C \cap \widetilde{\mathcal{RA}}_{t_{2}} \\
\llap{(\text{by Prop.~\ref{prop:monotonicity}})} 
& = & \displaystyle\bigcap_{\substack{\alpha < t_2}} \big[ \widetilde{\mathcal{RA}}_{\alpha}^C \cap \widetilde{\mathcal{RA}}_{t_{2}} \big] \\
& = & \big[ \displaystyle\bigcap_{\substack{\alpha < t_2}} \widetilde{\mathcal{RA}}_{\alpha}^C \big] \cap \widetilde{\mathcal{RA}}_{t_{2}}.\\
& = & \{(t,x): Q_{\min}(t,x)=t_2\}. \hfill\blacksquare\
\end{array}
\)\\

\phantomsection         
\label{apx:prop7}            
\noindent{\textbf{E.\;Proof of Proposition~\ref{prop:existence_absolute_continuity}}}:
\textit{Existence:} Let \( \bm{a} \in \mathcal{A}_t \), \( \bm{b} \in \mathcal{B}_t \), and \( (t, x, z) \in [0, T] \times \mathbb{R}^n \times\mathbb{R}\). By Carathéodory's existence theorem, there exists a unique absolutely continuous solution \( \phi^{\bm{a},\bm{b}}_{t, x} \) to the first \( n \) ODEs in~\eqref{eq:augmented_ode}. For the \((n+1)^{\text{th}}\) ODE, define
\(
\hat{f}(s) = \mathbf{1}_{\mathbb{C}_{2,s}^{C}}\big( \phi^{\bm{a}, \bm{b}}_{t, x}(s) \big).
\) \( \hat{f}(s) \) is measurable in \( s \) (see the proof of Proposition~\ref{prop5_soft}) and bounded by 1. Moreover, since \( \hat{f}(s) \) is independent of \( z \), it is trivially continuous in \( z \). Hence, by Carathéodory's existence theorem, there exists an absolutely continuous solution \( \xi^{\bm{a}, \bm{b}}_{t,x,z} \) to the \((n+1)^{\text{th}}\) ODE, yielding a solution \( \tilde{\phi}^{\bm{a}, \bm{b}}_{t,x,z} = \big( \phi^{\bm{a}, \bm{b}}_{t, x},\, \xi^{\bm{a}, \bm{b}}_{t,x,z} \big) \) to~\eqref{eq:augmented_ode}.\\
\noindent\textit{Uniqueness:} Suppose \( \tilde{\phi}_{t,x,z;1}^{\bm{a}, \bm{b}} \) and \( \tilde{\phi}_{t,x,z;2}^{\bm{a}, \bm{b}}\) are two distinct solutions to~\eqref{eq:augmented_ode}. By the existence proof, the first \( n \) ODEs have a unique solution \( \phi^{\bm{a}, \bm{b}}_{t,x} \), so both solutions have identical trajectories for the first \( n \) states:
\(
\tilde{\phi}_{t,x,z;i}^{\bm{a}, \bm{b}}(s) = \big( \phi^{\bm{a}, \bm{b}}_{t,x}(s),\; \xi_{t,x,z;i}^{\bm{a}, \bm{b}}(s) \big), \quad i = 1, 2.
\)
Each \(\xi_{t,x,z;i}^{\bm{a}, \bm{b}}(s)\) satisfies
\(
\xi_{t,x,z;i}^{\bm{a}, \bm{b}}(s) = z + \int_t^s \mathbf{1}_{\mathbb{C}_{2,\psi}^{C}} \big( \phi_{t,x}^{\bm{a}, \bm{b}}(\psi) \big)\, d\psi, \quad \forall s \in [t, T].
\)
Since the right-hand side is identical for both \( i = 1, 2 \), we conclude \( \xi_{t,x,z;1}^{\bm{a}, \bm{b}} = \xi_{t,x,z;2}^{\bm{a}, \bm{b}} \), and hence \( \tilde{\phi}_{t,x,z;1}^{\bm{a}, \bm{b}}= \tilde{\phi}_{t,x,z;2}^{\bm{a}, \bm{b}} \). \hfill$\blacksquare$\
\vspace{0.5em}

\phantomsection   
\label{apx:prop8}             
\noindent{\textbf{F.\;Proof of Proposition~\ref{introduction_of_lip_func}}}:
\textbf{1)} This is immediate from the definition of \(h_\epsilon(s, x)\). \\
\textbf{2)} Let \((s, x, a, b)\!\in\![0,T]\!\times\!\mathbb{R}^n\!\times\!\mathbb{A}\!\times\!\mathbb{B}\), and choose \(\epsilon_1, \epsilon_2 > 0\) with \(\epsilon_1 \!<\!\epsilon_2\), so that \(\tfrac{1}{\epsilon_2} \!<\!\tfrac{1}{\epsilon_1}\). Set \(d:= c_{2}(s,x)\) and consider:

\textbf{Case 1:} \(d \le 0\). Then \(d/\epsilon_i \le 0\) for \(i\in\{1,2\}\), hence
\(\max(0,d/\epsilon_2)=\max(0,d/\epsilon_1)\),
which implies
\(\min(1,\max(0,d/\epsilon_2))=\min(1,\max(0,d/\epsilon_1))\).

\textbf{Case 2:} \(d > 0\). Then \(\tfrac{d}{\epsilon_i} > 0\) for both \(i = 1,2\), and since \(\epsilon_1 < \epsilon_2\), it follows that \(\tfrac{d}{\epsilon_1} > \tfrac{d}{\epsilon_2}\). Hence,
\(
\min(1, \max(0, \tfrac{d}{\epsilon_2})) \leq \min(1, \max(0, \tfrac{d}{\epsilon_1}))
\). 
\noindent Hence, in both cases 1 and 2, \(h_{\epsilon_2}(s, x) \leq h_{\epsilon_1}(s, x)\), and since \(\tilde{f}_\epsilon(s, x, a, b)\) includes \(-h_\epsilon(s, x)\) as its last component, we have
\(
\tilde{f}_{\epsilon_1}(s, x, a, b) \leq \tilde{f}_{\epsilon_2}(s, x, a, b).
\) Therefore, the family \(\{\tilde{f}_\epsilon\}_{\epsilon > 0}\) is monotone decreasing as \(\epsilon \to 0^+\).\\
\noindent Next, we establish the pointwise convergence of \( \tilde{f}_\epsilon \) as \( \epsilon \to 0^+ \).
\(
\lim_{\epsilon \to 0^+} \tilde{f}_{\epsilon}(s, x, a, b) =
\begin{bmatrix}
f(s, x, a, b) \\
-\lim_{\epsilon \to 0^+} h_{\epsilon}(s, x)
\end{bmatrix}.
\)\\
\noindent Hence, it suffices to consider the limit of \( h_{\epsilon} \) as \( \epsilon \to 0^+ \). The function \( h_{\epsilon} \) is continuous with respect to \( \epsilon \), which allows the limit to be taken directly through the function. This yields:\\
\(
\begin{aligned}
\lim_{\epsilon\to0^+} h_\epsilon(s,x)
&=\min\!\Bigl\{1,\max\!\Bigl(0,\lim_{\epsilon\to0^+}\tfrac{1}{\epsilon}\,c_{2}(s,x)\Bigr)\Bigr\}\\[0.0pt]
&=\begin{cases}
1 & \text{if } c_{2}(s,x)>0 \;\bigl(\tfrac{1}{\epsilon}c_{2}(s,x)\!\to\!+\infty\bigr),\\[4pt]
0 & \text{if } c_{2}(s,x)=0 \;\bigl(\tfrac{1}{\epsilon}c_{2}(s,x)=0\bigr),\\[4pt]
0 & \text{if } c_{2}(s,x)<0 \;\bigl(\tfrac{1}{\epsilon}c_{2}(s,x)\!\to\!-\infty\bigr).
\end{cases}
\end{aligned}
\)\\
\noindent Therefore, if \( x \in \mathbb{C}_{2,s} \) (i.e., \( c_{2}(s,x) \leq 0 \)), the limit of \( h_{\epsilon}(s, x) \) is \( 0 \). If \( x \notin \mathbb{C}_{2,s} \) (i.e., \( c_{2}(s,x) > 0 \)), the limit of \( h_{\epsilon}(s, x) \) is \( 1 \). This proves that \( h_{\epsilon}(s, x) \) converges to
\(
\mathbf{1}_{\mathbb{C}_{2,s}^c}(x).
\)
Substituting into the augmented dynamics \( \tilde{f}_{\epsilon} \), the limit is:\\
\(
\lim_{\epsilon \to 0^+} \tilde{f}_{\epsilon}(s, x, a, b) = 
\begin{bmatrix}
f(s, x, a, b) \\
-\mathbf{1}_{\mathbb{C}_{2,s}^c}(x)
\end{bmatrix}.
\)\\
\textbf{3)} The following lemma is used in the proof.
\begin{lemma}\label{lemma1}
Let \( A, B \subset \mathbb{R}^n \) be two compact sets. Let \( d_A(x) \) and \( d_B(x) \) denote the signed distance functions associated with \( A \) and \( B \), respectively, and let \( d_H(A, B) \) denote the Hausdorff distance between \( A \) and \( B \). For any \(x \in \mathbb{R}^n\), the signed-distance functions satisfy \(|d_A(x)-d_B(x)| \le d_H(A,B)\), where the Hausdorff distance is defined by  
\(d_H(A,B)=\max\{\sup_{a\in A}\inf_{b\in B}\|a-b\|,\;\sup_{b\in B}\inf_{a\in A}\|a-b\|\}\).
\end{lemma}

The result of the lemma follows directly from the definitions of the Hausdorff distance and the signed distance function. For a detailed proof, see Theorem 1 in \cite{kraft2018hausdorff}. \\
\indent Since \(\mathbb{C}_{2,t}\) has uniformly bounded temporal variation, (as assumed in Section~\ref{sec:constraint_sets}), it follows that
\(
\exists\,L_v>0 \text{ such that }\)
\(
d_H(\mathbb{C}_{2,s}, \mathbb{C}_{2,s'}) \leq L_v |s - s'|, \quad \forall s, s' \in [0, T].
\) From Lemma~\ref{lemma1}, this implies a corresponding bound on the signed distance function \(c_{2}\), given by: \\
\(
|c_{2}(s,x) - c_{2}(s',x)| \leq L_v |s - s'|, \quad \forall x \in \mathbb{R}^n, \, \forall s, s' \in [0, T].
\)\\ 
Therefore, \(c_{2}\) is Lipschitz continuous with respect to time.\\
\indent Let \(x_1, x_2 \in \mathbb{R}^n\) and \(s_1, s_2 \in [0, T]\). By definition of the signed–distance function, \(c_{2}(s,\cdot)\) is \(1\)-Lipschitz in \(x\); in particular, for all \(s\in[0,T]\) \\
\(
\left| c_{2}(s,x_1) - c_{2}(s,x_2) \right| \le \|x_1 - x_2\|.
\) Using the triangle inequality on \( |h_{\epsilon}(s_1, x_1) - h_{\epsilon}(s_2, x_2)| \), we get:\\
\(
\begin{aligned}
\left| h_{\epsilon}(s_1, x_1) - h_{\epsilon}(s_2, x_2) \right| 
&\leq \left| h_{\epsilon}(s_1, x_1) - h_{\epsilon}(s_1, x_2) \right| \\
&\quad + \left| h_{\epsilon}(s_1, x_2) - h_{\epsilon}(s_2, x_2) \right|.
\end{aligned}
\)\\
\noindent For any \( u, v \in \mathbb{R} \), the functions \( \max(0, \cdot) \) and \( \min(1, \cdot) \) are 1-Lipschitz, so:
\(
\begin{aligned}
|\max(0, u) - \max(0, v)| &\leq |u - v|,
\end{aligned}
\)
\(
\begin{aligned}
|\min(1, u) - \min(1, v)| &\leq |u - v|.
\end{aligned}
\)Using these bounds and the Lipschitz continuity of \(c_{2}\), each term is bounded as:
\begingroup           
\setlength{\abovedisplayskip}{0.5pt}
\setlength{\belowdisplayskip}{0.5pt}
\[
\begin{aligned}
|h_{\epsilon}(s_1,x_1)-h_{\epsilon}(s_1,x_2)|
      &\le \frac{1}{\epsilon}\,\|x_1-x_2\|,\\
|h_{\epsilon}(s_1,x_2)-h_{\epsilon}(s_2,x_2)|
      &\le \frac{L_v}{\epsilon}\,|s_1-s_2|.
\end{aligned}
\]
\endgroup  
Combining these bounds leads to:\\
\(
\left| h_{\epsilon}(s_1, x_1) - h_{\epsilon}(s_2, x_2) \right| \leq \frac{1}{\epsilon} \|x_1 - x_2\| + \frac{L_v}{\epsilon} |s_1 - s_2|.
\)\\
\noindent
Thus \(h_\epsilon\) is Lipschitz in \((s,x)\). Since \(f\) is uniformly continuous, bounded, and uniformly Lipschitz in \(x\) on \([0,T]\times\mathbb{A}\times\mathbb{B}\), \(\tilde f_\epsilon\) is, for every \(\epsilon>0\), bounded and uniformly continuous, and Lipschitz in \((x,z)\) (uniformly in \((s,a,b)\)).\hfill$\blacksquare$\ 

\vspace{0.3em}
\phantomsection 
\label{apx:thm2}           
\noindent{\textbf{G.\;Proof of Theorem~\ref{family_value_functions}}}:
\noindent \textbf{1)} For every \( \epsilon>0 \), \(\tilde{f}_{\epsilon}(t,x,z,a,b)\) is bounded, uniformly
continuous, and Lipschitz in \(x,z\), uniformly over
\((t,a,b)\!\in\![0,T]\times\mathbb{A}\times\mathbb{B}\) (see Proposition~\ref{introduction_of_lip_func}). Since \(c_1(t,x),\,c_2(t,x)\), and \(g(t,x)\) are Lipschitz, it follows from\cite{Evans} that \(W_{\epsilon}(t,x,z)\) is Lipschitz continuous in \(t,x,z\).\\
\noindent \textbf{2)} Let \(0\!<\!\epsilon_{1}\!<\!\epsilon_{2}\). For any \((t, x, z)\!\in\![0, T]\!\times\!\mathbb{R}^n\times\!\mathbb{R}\), \(\bm{a}\!\in\!\mathcal{A}_{t}\), and \(\delta\!\in\!\Delta_{t}\), denote the unique Lipschitz continuous trajectories that solves the dynamics \(\tilde{f}_{\epsilon_{1}}\) and \(\tilde{f}_{\epsilon_{2}}\) by
\begingroup
\setlength{\abovedisplayskip}{0.5pt}
\setlength{\belowdisplayskip}{0.5pt}
\[
\begin{aligned}
\tilde{\phi}^{\bm{a}, \delta[\bm{a}]}_{t, x, z}(s) 
&= \big( \phi^{\bm{a}, \delta[\bm{a}]}_{ t, x}(s),\ 
\xi^{\bm{a}, \delta[\bm{a}]}_{\epsilon_i, t, x, z}(s) \big), 
\quad i = 1, 2
\end{aligned}
\]
\endgroup
\noindent Since \(h_{\epsilon_{1}}\!\geq\!\!h_{\epsilon_{2}}\), it follows that \(\forall\!s\!\in\![t, T]\): 
\(
\xi^{\bm{a}, \delta[\bm{a}]}_{\epsilon_1, t, x, z}(s)\!\leq\!\xi^{\bm{a}, \delta[\bm{a}]}_{\epsilon_2, t, x, z}(s).
\) Consider the following functional 
\begingroup
\setlength{\abovedisplayskip}{-1.5pt}
\setlength{\belowdisplayskip}{0.5pt}
\begin{equation*}
\begin{aligned}
J_\epsilon(t, x, z, \bm{a}, \delta[\bm{a}]) 
= \min_{\tau \in [t, T]} \max \Big\{ &
\max_{s \in [t, \tau]} c_1\big(s, \phi^{\bm{a}, \delta[\bm{a}]}_{t, x}(s) \big), \\
& \hspace*{-13.2em} c_2\big(\tau, \phi^{\bm{a}, \delta[\bm{a}]}_{t, x}(\tau)\big),\; -\xi^{\bm{a}, \delta[\bm{a}]}_{\epsilon, t, x, z}(\tau),\;
g\big(\tau, \phi^{\bm{a}, \delta[\bm{a}]}_{t, x}(\tau)\big) \Big\}
\end{aligned}
\end{equation*}
\endgroup
\noindent Hence: 
\(
J_{\epsilon_{1}}(t, x, z, \bm{a}, \delta[\bm{a}]) \geq J_{\epsilon_{2}}(t, x, z, \bm{a}, \delta[\bm{a}])
\).
Taking the supremum over all admissible strategies \(\delta\) and the infimum over all control signals \(\bm{a}\), the following holds:
\begingroup
\setlength{\abovedisplayskip}{0.1pt}
\setlength{\belowdisplayskip}{0.1pt}
\[
\begin{aligned}
W_{\epsilon_1}(t,x,z)
  &= \sup_{\delta\in\Delta_t}\,\inf_{\bm{a}\in\mathcal{A}_t}
     J_{\epsilon_{1}}(t,x,z,\bm{a},\delta[\bm{a}])\\
  &\ge \sup_{\delta\in\Delta_t}\,\inf_{\bm{a}\in\mathcal{A}_t}
        J_{\epsilon_{2}}(t,x,z,\bm{a},\delta[\bm{a}]) = W_{\epsilon_2}(t,x,z).
\end{aligned}
\]
\endgroup
\noindent This shows that \(W_{\epsilon}\) increases monotonically as \( \epsilon \to 0^+ \).\\
\noindent \textbf{5)} Fix \((t,x,z)\!\in[0,T]\times\mathbb{R}^n\times\mathbb{R}\) and set \(\tilde x=(x,z)\).
For any controls \(\bm a\in\mathcal A_t\), \(\bm b\in\mathcal B_t\): \\
\noindent\textit{\textbullet}\;
By Proposition~\ref{prop:existence_absolute_continuity} the augmented
ODE, \(\tilde f\), admits a unique, absolutely continuous trajectory
\(
\tilde\phi_{t,\tilde x}^{\bm a,\bm b}(\cdot)
  =\bigl(\phi_{t,x}^{\bm a,\bm b}(\cdot),
         \xi_{t,\tilde x}^{\bm a,\bm b}(\cdot)\bigr):
   [t,T]\!\to\!\mathbb{R}^n\times\mathbb{R}.
\)\\
\noindent\textit{\textbullet}
For each \(\epsilon>0\), let
\(\tilde\phi_{\epsilon,t,\tilde x}^{\bm a,\bm b}:[t,T]\to
  \mathbb{R}^n\times\mathbb{R}\)
be the solution of \(\tilde f_\epsilon\).
Proposition~\ref{introduction_of_lip_func} shows that
\(\tilde f_\epsilon\) is uniformly continuous in
\((s,\tilde x,a,b)\) and globally Lipschitz in \(\tilde x\)
(uniform in \(s,a,b\)); hence, by Carathéodory’s existence theorem, there exists a unique Lipschitz continuous trajectory, which can be decomposed as
\(
\tilde\phi_{\epsilon,t,\tilde x}^{\bm a,\bm b}
  =(\phi_{t,x}^{\bm a,\bm b},
    \xi_{\epsilon,t,\tilde x}^{\bm a,\bm b}),
\)
\emph{where
\(\xi_{\epsilon,t,\tilde x}^{\bm a,\bm b}\) satisfies the last ODE of
\(\tilde f_\epsilon\) almost everywhere on \([t,T]\).}
\noindent\begin{lemma}\label{lem:infimum_difference_bound}\;
\emph{For real functions \(X,Y:D\!\to\!\mathbb{R}\),}
{\setlength{\abovedisplayskip}{-1pt}
 \setlength{\belowdisplayskip}{-1pt}
 \[
   \bigl|\inf_{x\in D}X(x)-\inf_{x\in D}Y(x)\bigr|
      \le \sup_{x\in D}|X(x)-Y(x)|
 \]}
 \end{lemma}
\noindent\textit{Proof.}\;
For every \(x\in D\) we have \(\inf_D X \le X(x)\) and
\(\inf_D Y \le Y(x)\).
Subtracting and taking absolute values gives
\(\lvert \inf_D X - \inf_D Y\rvert
  \le |X(x)-Y(x)|\);
taking the supremum over \(x\) proves the claim. \\ 
For any strategy \(\!\delta\!\in\!\Delta_t\), control signal \(\!\bm a\!\in\!\mathcal A_t\), and \(\!\epsilon\!>\!0\), let
{\setlength{\abovedisplayskip}{4pt}%
 \setlength{\abovedisplayshortskip}{4pt}%
 \setlength{\belowdisplayskip}{2pt}%
 \setlength{\belowdisplayshortskip}{2pt}%
 \setlength{\jot}{2pt}%
 \[
   \xi_\kappa(s):=
   \begin{cases}
     \xi_{t,\tilde x}^{\bm a,\delta[\bm a]}(s), & \kappa=0,\\
     \xi_{\epsilon,t,\tilde x}^{\bm a,\delta[\bm a]}(s), & \kappa=\epsilon ,
   \end{cases}
   \qquad
   \kappa\in\{0,\epsilon\}.
 \]}
Define $F_\kappa(\tau,\bm a,\delta[\bm a])$, $J_\kappa(t,x,z,\bm a,\delta[\bm a])$,
and $\Delta_J$ as
{\setlength{\abovedisplayskip}{4pt}
 \setlength{\abovedisplayshortskip}{4pt}
 \setlength{\belowdisplayskip}{2pt}
 \setlength{\belowdisplayshortskip}{2pt}
 \setlength{\jot}{2pt}
 \begin{gather*}
   F_\kappa(\tau,\bm a,\delta[\bm a])
   := \max\!\Bigl\{
\max_{s\in[t,\tau]}
          c_1\!\bigl(s,\phi_{t,x}^{\bm a,\delta[\bm a]}(s)\bigr),\\[-2pt]
          c_2\!\bigl(\tau,\phi_{t,x}^{\bm a,\delta[\bm a]}(\tau)\bigr),
          -\xi_\kappa(\tau),\;
        g\!\bigl(\tau,\phi_{t,x}^{\bm a,\delta[\bm a]}(\tau)\bigr)
      \Bigr\},\\[3pt]
J_\kappa(t,x,z,\bm a,\delta[\bm a])
 := \min_{\tau\in[t,T]} F_\kappa(\tau,\bm a,\delta[\bm a]),~
\Delta_J \!:= \!\lvert J_0 - J_\epsilon\rvert.
 \end{gather*}}
 \(
\begin{aligned}
\Delta_J &= \Bigl|
     \mathop{\min}\limits_{\tau\in[t,T]}
       F_0(\tau,\bm a,\delta[\bm a]) -
     \mathop{\min}\limits_{\tau\in[t,T]}
       F_\epsilon(\tau,\bm a,\delta[\bm a])
   \Bigr| \\[-0.5pt]
&\le \mathop{\sup}\limits_{\tau\in[t,T]}
      \bigl|F_0(\tau,\bm a,\delta[\bm a]) -
            F_\epsilon(\tau,\bm a,\delta[\bm a])\bigr|      
      \;\text{\scriptsize(by Lemma~\ref{lem:infimum_difference_bound})} \\[-0.5pt]
&\le \mathop{\sup}\limits_{\tau\in[t,T]}
      \bigl|\xi^{\bm a,\delta[\bm a]}_{t,\tilde x}(\tau) -
            \xi^{\bm a,\delta[\bm a]}_{\epsilon,t,\tilde x}(\tau)\bigr|
     \;(\text{\scriptsize$\substack{\max\ \text{operator is}\\[-1pt]1\text{-Lipschitz})}$}) \\[-0.5pt]
&= \bigl|\xi^{\bm a,\delta[\bm a]}_{t,\tilde x}(T) -
          \xi^{\bm a,\delta[\bm a]}_{\epsilon,t,\tilde x}(T)\bigr| 
\;\text{\scriptsize$
\!\!  \bigl(
    \substack{|\xi^{\bm a,\delta[\bm a]}_{t,\tilde x}(\tau)-
               \xi^{\bm a,\delta[\bm a]}_{\epsilon,t,\tilde x}(\tau)|\ \text{is continuous}\\[-1pt]\text{\ in}~\tau~
               \text{and increasing on }[t,T]}\bigr)$}\\[-0.5pt]
&\le \int_t^T\!
      \Bigl|
        \mathbf 1_{\mathbb C_{2,s}^C}
          \bigl(\phi^{\bm a,\delta[\bm a]}_{t,x}(s)\bigr)\\[-4.8pt]
        &\quad\quad\; -\min\!\Bigl[1,
           \max\!\bigl(0,
             \tfrac{1}{\epsilon}\,
             c_{2}(s, 
               \phi^{\bm a,\delta[\bm a]}_{t,x}(s)
           \bigr)\Bigr]
      \Bigr|\,ds.
\;\text{\scriptsize$
  \bigl(
    \substack{\text{by triangle}\\[-0.5pt]\text{inequality}}
  \bigr)
$}
\end{aligned}
\)\\
\noindent Define the set \( E_{\epsilon, t, x}(\bm{a}, \delta[\bm{a}]) \subseteq [t,T] \) as the collection of times \( s \) for which \( \phi^{\bm{a}, \delta[\bm{a}]}_{t, x}(s) \) lies within an\;\( \epsilon \)-neighborhood of the boundary of \( \mathbb{C}_{2,s} \), i.e., \( E_{\epsilon,t,x}(\bm{a},\delta[\bm{a}])=\{ s\in[t,T] : 0<c_{2}(s,\phi^{\bm{a},\delta[\bm{a}]}_{t,x}(s))<\epsilon \} \). This set is open and measurable as it is the preimage of the open interval \( (0,\epsilon)\subseteq \mathbb{R} \) under the measurable mapping \( s\mapsto c_{2}(s,\phi^{\bm{a},\delta[\bm{a}]}_{t,x}(s))\). The functions 
\( \mathbf{1}_{\mathbb{C}_{2,s}^C}(\phi^{\bm{a},\delta[\bm{a}]}_{t,x}(s)) \) and 
\( \min(1,\max(0,\tfrac{1}{\epsilon}c_{2}(s,\phi^{\bm{a},\delta[\bm{a}]}_{t,x}(s)))) \) only differ on this set. Thus, the integral bound on \(\Delta_J\) reduces to
\(\int_{E_{\epsilon,t,x}(\bm{a},\delta[\bm{a}])} \bigl|1 - \min\bigl(1,\max(0,\tfrac{1}{\epsilon} c_{2}(s,\phi^{\bm{a},\delta[\bm{a}]}_{t,x}(s)))\bigr)\bigr|\,ds\).  
Since the integrand is bounded by \(1\), it follows that  
\(\bigl|J_{0}(t, x, z, \bm{a}, \delta[\bm{a}]) - J_\epsilon(t, x, z, \bm{a}, \delta[\bm{a}])\bigr| \leq \mu\bigl(E_{\epsilon, t, x}(\bm{a}, \delta[\bm{a}])\bigr)\),  
where \(\mu(\cdot)\) denotes the Lebesgue measure defined on \([0,T]\). Given that \(J_{0}(t, x, z, \bm{a}, \delta[\bm{a}]) \geq J_{\epsilon}(t, x, z, \bm{a}, \delta[\bm{a}])\), we obtain:  
\(
\begin{aligned}
 \big|\sup_{\delta \in \Delta_{t}} \inf_{\bm{a} \in \mathcal{A}_{t}} 
 J_{0}(t, x, z, \bm{a}, \delta[\bm{a}]) 
 - \sup_{\delta \in \Delta_{t}} \inf_{\bm{a} \in \mathcal{A}_{t}} 
 J_\epsilon(t, x, z, \bm{a}, \delta[\bm{a}]) \big| 
&\\
\leq \sup_{\delta \in \Delta_{t}} \inf_{\bm{a} \in \mathcal{A}_{t}} 
\mu\big(E_{\epsilon, t, x}(\bm{a}, \delta[\bm{a}])\big),
\end{aligned}
\)
where \(\mu\bigl(E_{\epsilon, t, x}(\bm{a}, \delta[\bm{a}])\bigr) = \int_t^T \mathbf{1}_{E_{\epsilon, t, x}(\bm{a}, \delta[\bm{a}])}(s)\,ds\).\\
\noindent Because \(f(s,x,a,b)\) is uniformly continuous and Lipschitz in \(x\) (uniformly in \((s,a,b)\in[0,T]\times\mathbb{A}\times\mathbb{B}\)), hence the trajectory \(\phi^{\bm{a},\delta[\bm{a}]}_{t,x}(s)\) is uniformly continuous in time \(s\). Let \((\bm a_n,\delta_n)_{n\in\mathbb N}\) be a sequence with \(\bm a_n \to \bm a_0\) in \(L^\infty([t,T];\mathbb A)\), and set \(\bm b_n := \delta_n[\bm a_n]\), \(\bm b_0 := \delta_0[\bm a_0]\), with \(\bm b_n \to \bm b_0\) in \(L^\infty([t,T];\mathbb B)\). By Grönwall’s inequality and the Dominated Convergence Theorem,
\(
\phi^{\bm a_n,\delta_n[\bm a_n]}_{t,x}(\cdot)
\)
converges uniformly on \([t,T]\) to
\(
\phi^{\bm a_0,\delta_0[\bm a_0]}_{t,x}(\cdot).
\)
As shown in the proof of Proposition~\ref{introduction_of_lip_func}, \(c_{2}(s,x)\) is
Lipschitz in \((s,x)\). Hence the compositions
\(
s\mapsto c_2\!\big(s,\phi^{\bm a_n,\delta_n[\bm a_n]}_{t,x}(s)\big)
\)
converge uniformly to
\(
s\mapsto c_2\!\big(s,\phi^{\bm a_0,\delta_0[\bm a_0]}_{t,x}(s)\big).
\) Fix \(s\in E_{\epsilon,t,x}(\bm a_0,\delta_0)\). Then there exists \(N\) such that for all \(n\ge N\), \(s\in E_{\epsilon,t,x}(\bm a_n,\delta_n)\). Thus, \( \mathop{\liminf}\limits_{n\to\infty} \mathbf{1}_{E_{\epsilon,t,x}(\bm a_n,\delta_n)}(s) \ge \mathbf{1}_{E_{\epsilon,t,x}(\bm a_0,\delta_0)}(s)=1 \). Conversely, if \( s \notin E_{\epsilon,t,x}(\bm a_0,\delta_0) \), then  \( \mathop{\liminf}\limits_{n\to\infty} \mathbf{1}_{E_{\epsilon,t,x}(\bm a_n,\delta_n)}(s) \ge \mathbf{1}_{E_{\epsilon,t,x}(\bm a_0,\delta_0)}(s)=0 \).
\noindent Combining both cases yields, for each \(s \in [t, T]\),
\begingroup\setlength{\abovedisplayskip}{2pt}\setlength{\belowdisplayskip}{2pt}
\(
\mathop{\liminf}\limits_{n\to\infty} \mathbf{1}_{E_{\epsilon,t,x}(\bm{a}_{n}, \delta_{n})}(s) \geq \mathbf{1}_{E_{\epsilon,t,x}(\bm{a}_{0}, \delta_{0})}(s)
\)
\endgroup
\noindent which shows that, for each fixed \(s\in[t,T]\), the map
\((\bm{a},\delta)\mapsto \mathbf{1}_{E_{\epsilon,t,x}(\bm{a},\delta)}(s)\)
is lower semicontinuous. A similar argument shows that \(\displaystyle \operatorname*{inf}_{\bm{a} \in \mathcal{A}_t} \mathbf{1}_{E_{\epsilon,t,x}(\bm{a}, \delta)}(s)\) is also lower semi-continuous in \(\delta\) for each fixed \(s \in [t, T]\). Since
\(\mathbf{1}_{E_{\epsilon,t,x}(\bm{a}, \delta)}(s)\) and \(\displaystyle \operatorname*{inf}_{\bm{a} \in \mathcal{A}_t} \mathbf{1}_{E_{\epsilon,t,x}(\bm{a}, \delta)}(s)\) are normal integrands (i.e., lower semicontinuous and measurable), $\Delta_t$ and $\mathcal{A}_t$ are decomposable spaces (i.e., if \(f,\,g\in \Delta_t\) (or \(\mathcal{A}_t\)) and \(A\) is any measurable subset of \([t, T]\), then \(
h(\omega) := f(\omega)\, \mathbf{1}_A(\omega) + g(\omega)\, \mathbf{1}_{A^c}(\omega)
\) belongs to \(\Delta_t\) (or \(\mathcal{A}_t\))), and $\mathbf{1}_{E_{\epsilon,t,x}(\bm{a},\delta)}(s)$ is bounded for every \(\bm{a}\in\mathcal{A}_t\) and \(\delta\in\Delta_t\), all conditions of Rockafellar's Interchange Theorem~\cite{Rockfellar} hold. Therefore, the supremum and infimum operations may be interchanged with the integral.\\
\(
\begin{aligned}
\sup_{\mathclap{\delta \in \Delta_t}}~\inf_{\mathclap{\bm{a} \in \mathcal{A}_t}}
\mu\Bigl(E_{\epsilon,t,x}(\bm{a},\delta[\bm{a}])\Bigr)
&=\!\!\int_t^T\!\!\!\sup_{\mathclap{\delta \in \Delta_t}}~ \inf_{\mathclap{\bm{a}\in\mathcal{A}_t}}
\mathbf{1}_{E_{\epsilon,t,x}(\bm{a},\delta[\bm{a}])}(s)\,ds \\
&=\!\!\int_t^T\!\!\!\mathbf{1}_{\bigcup\limits_{\delta\in\Delta_t}\!\bigcap\limits_{\bm{a}\in\mathcal{A}_t}E_{\epsilon,t,x}(\bm{a},\delta[\bm{a}])}(s)\,ds \\
&=\;\mu\Bigl(\bigcup_{\mathclap{\delta\in\Delta_t}}\,\bigcap_{\bm{a}\in\mathcal{A}_t}
E_{\epsilon,t,x}(\bm{a},\delta[\bm{a}])\Bigr) \\[-1.4ex]
\llap{\(\fontsize{7}{8}\selectfont
\begin{array}{c}
\text{(since }E_{\epsilon,t,x}(\bm{a},\delta[\bm{a}])\subseteq [t,T] \\[1mm]
\forall\,\bm{a}\in\mathcal{A}_t,\ \forall\,\delta\in\Delta_t\text{)}
\end{array}\!\)}
&\le \mu([t,T]) = T-t
\end{aligned}
\)\\
\noindent Hence, it directly follows that: \\
 \(
\begin{aligned}
\big| W(t, x, z) - W_\epsilon(t, x, z) \big| 
&\leq \mu\Bigl(\bigcup_{\mathclap{\delta\in\Delta_t}}\,\bigcap_{\bm{a}\in\mathcal{A}_t}
E_{\epsilon,t,x}(\bm{a},\delta[\bm{a}])\Bigr) \\
&\leq T - t.
\end{aligned}
\)\\
\textbf{3)} Fix \( (t,x) \in [0,T] \times \mathbb{R}^n \). 
For every \(\!\epsilon\!>\!0\!\) we have  
\(\mu\bigl(
      \mathop{\bigcup}\limits_{\delta \in \Delta_t}\!
      \mathop{\bigcap}\limits_{\bm{a} \in \mathcal{A}_t}\!
      E_{\epsilon,t,x}(\bm{a},\delta[\bm{a}])
\bigr)
\le
\mu\bigl(
      \mathop{\bigcup}\limits_{\delta \in \Delta_t}
      E_{\epsilon,t,x}(\bm{a},\delta[\bm{a}])
\bigr).\) Since each \(E_{\epsilon,t,x}(\bm{a},\delta[\bm{a}])\subset [0,T]\) has finite measure, and for any \(\bm{a} \in \mathcal{A}_t\) the family 
\(\mathop{\bigcup}\limits_{\delta \in \Delta_t} E_{\epsilon,t,x}(\bm{a},\delta[\bm{a}])\) decreases as \(\epsilon \downarrow 0\), the continuity-from-above property of the Lebesgue measure gives  
\(\displaystyle
\lim_{\epsilon \to 0^{+}}
\mu\bigl(
      \mathop{\bigcup}\limits_{\delta \in \Delta_t}
      E_{\epsilon,t,x}(\bm{a},\delta[\bm{a}])
\bigr)
=
\mu\bigl(
      \mathop{\bigcap}\limits_{\epsilon > 0}
      \mathop{\bigcup}\limits_{\delta \in \Delta_t}
      E_{\epsilon,t,x}(\bm{a},\delta[\bm{a}])
\bigr).\) Since it is clear that
\(\mathop{\bigcap}\limits_{\epsilon > 0}
  \mathop{\bigcup}\limits_{\delta \in \Delta_t}
  E_{\epsilon,t,x}(\bm{a},\delta[\bm{a}])\!=\!\varnothing,\)  
we obtain  
\(\displaystyle
\lim_{\epsilon \to 0^{+}}
\mu\bigl(
      \mathop{\bigcup}\limits_{\delta \in \Delta_t}\!
      \mathop{\bigcap}\limits_{\bm{a} \in \mathcal{A}_t}\!
      E_{\epsilon,t,x}(\bm{a},\delta[\bm{a}])
\bigr) = 0.\)\\
\noindent Therefore, by part \textbf{5)}, \(W_{\epsilon}\) converges to \(W\) pointwise as \(\epsilon \downarrow 0\).\\
\noindent \textbf{4)} For every $\epsilon>0$, $W_\epsilon$ is Lipschitz continuous by part~\textbf{1)}, hence measurable. Since $W_\epsilon$ is finite on $[0,T]\times\mathbb{R}^n\times\mathbb{R}$ and converges pointwise to $W$, hence $W$ is also finite and measurable as the pointwise limit of finite and measurable functions. Let $D\!\subset\![0,T]\!\times\!\mathbb{R}^n\!\times\!\mathbb{R}$ be compact; then $\mu(D)\!<\!\infty$. By Egorov's Theorem~\cite{Royden}, for any $\eta\!>\!0$, there exists a measurable set $\tilde{D}\!\subseteq\!D$ with $\mu(D\!\setminus\!\tilde{D})<\eta$ on which the convergence $W_\epsilon\to W$ is uniform. Define \(
A_\epsilon := \{(t,x,z)\in D : |W(t,x,z)-W_\epsilon(t,x,z)|\!>\!\eta\}.
\) Then we have \(
\mu(A_\epsilon)=\mu(A_\epsilon\cap\tilde{D})+\mu(A_\epsilon\cap(D\!\setminus\!\tilde{D})).
\) Due to uniform convergence on $\tilde{D}$, there exists $\epsilon_0>0$ such that for all $\!\epsilon\!\leq\!\epsilon_0$, 
\(
|\!W(t,x,z)\!-\!W_\epsilon(t,x,z)\!|\!<\!\eta~\forall (t,x,z)\in\!\tilde{D},
\)
hence $A_\epsilon\!\cap\!\tilde{D}\!=\!\emptyset$. Thus, for every $\epsilon\!\leq\!\epsilon_0$, \(\!
\mu(A_\epsilon)\!=\!\mu(D\!\setminus\!\tilde{D})<\!\eta.\) Since $\!\eta\!>\!0$ is arbitrary, we have $W_\epsilon\!\to\!W$ in measure on $D$.\hfill$\blacksquare$\

\vspace{0.29em}
\phantomsection
\label{apx:thm3}           
\!\!\!\!\!\!\noindent{\textbf{H.\;Proof of Proposition~\ref{thm2}}}:
\noindent For \(t\in[0,T]\), \((x,z)\in\mathbb R^{n}\times\mathbb R\) and \(\epsilon>0\), set \(c(t,x,z):=\max\!\bigl(c_2(t,x),-z,g(t,x)\bigr)\).
Because each argument of the \(\max\) is Lipschitz continuous, so is \(c\). Using \(c\) the value function becomes
\(
\begin{aligned}
W_\epsilon(t, x, z)&=\sup_{\delta\in\Delta_t}\inf_{a\in\mathcal{A}_t}\min_{\tau\in [t,T]}\max\Big\{\!\max_{s\in[t,\tau]} c_1\big(s,\phi^{\bm{a},\delta[\bm{a}]}_{t,x}(s)\big),\\
&\quad c\big(\tau,\phi^{\bm{a},\delta[\bm{a}]}_{t,x}(\tau),\xi^{\bm{a},\delta[\bm{a}]}_{\epsilon,t,x,z}(\tau)\big)\!\Big\}.
\end{aligned}
\)
\noindent Therefore, by Lemma~2 in~\cite{c3}, for all \( t\!\in\![0, T)\), \(\tilde{\delta}\!>0\) with \( t+\tilde{\delta}\!\leq T \), and \( (x, z)\!\in\!\mathbb{R}^n\!\times\!\mathbb{R}\), \( W_\epsilon(t, x, z) \) satisfies:
\(
\begin{aligned}
W_\epsilon(t,x,z)
  &=\!\sup_{\delta\in\Delta_t}\!\inf_{a\in\mathcal{A}_t} \Big\{\!\min \Big[ \min_{\tau\in[t,t+\tilde{\delta}]} \max \Big(\\[-0.1ex]
& \hspace{-0.9em}\max_{s\in[t,\tau]} c_1\bigl(s, \phi^{\bm{a},\delta[\bm{a}]}_{t,x}(s) \bigr),\,
   c\bigl( \tau, \phi^{\bm{a},\delta[\bm{a}]}_{t,x}(\tau), \xi^{\bm{a},\delta[\bm{a}]}_{\epsilon,t,x,z}(\tau\!)\bigr)
   \Big),\\[-0.1ex]
&\hspace{-0.5em} \max \Big(
W_\epsilon\bigl( t+\tilde{\delta}, \phi^{\bm{a},\delta[\bm{a}]}_{t,x}(t+\tilde{\delta}),
                   \xi^{\bm{a},\delta[\bm{a}]}_{\epsilon,t,x,z}(t+\tilde{\delta}) \bigr),\\[-0.1ex]
&\quad
   \max_{\tau\in[t,t+\tilde{\delta}]} c_1\bigl(\tau, \phi^{\bm{a},\delta[\bm{a}]}_{t,x}(\tau) \bigr)
   \Big) \Big] \Big\}
\end{aligned}
\)\\
\noindent which is equivalent to~\eqref{eq:dynamic_prog}. This proves Proposition~\ref{thm2}.\hfill$\blacksquare$\\\
\vspace{0.095em}
\noindent{\textbf{Proof of Theorem~\ref{final_theorem}}}:
With \((c_2,-z,g)\) replaced by the function \(c\), the construction of \(W_\epsilon\) in
Proposition~\ref{thm2} places the problem within the scope of Theorem~1
in~\cite{c3}. Following the same proof as Theorem~1 in~\cite{c3}, we conclude that \(W_\epsilon\) is the unique viscosity
solution of the resulting HJI–VI. \hfill\(\blacksquare\)

\vspace{0.095em}
\phantomsection   
\label{apx:prop9}           
\noindent{\textbf{I.\;Proof of Proposition~\ref{Prop_9}}}:
The proof is analogous to the~\hyperref[apx:prop4]{proof of Proposition~\ref{prop:RA_{Q} characterization}} and, as such, is omitted. \hfill\(\blacksquare\)

\phantomsection      
\label{apx:prop10}           
\noindent{\textbf{J.\;Proof of Theorem~\ref{Prop_10}}}:
\textbf{1)}
    For any \(\!\epsilon\!>\!0\), let \((t, x)\in\widetilde{\mathcal{RA}}_{Q}^{\epsilon}\) and \(Q\in(0, T]\). Define \( b := \mu\!\left(
  \bigcup_{\delta\in\Delta_t}
  \bigcap_{\bm a\in\mathcal A_t}
  E_{\epsilon,t,x}\!\big(\bm a,\delta[\bm a]\big)
\right) \). By definition, \(
    W_\epsilon(t,x,Q)\leq-b.\) From Theorem~\ref{family_value_functions},
    \(
    W(t,x,Q)\leq W_\epsilon(t,x,Q)+b\leq0.
    \) Thus, \(W(t,x,Q)\leq 0\); therefore, \((t, x) \in \mathcal{RA}_Q\). Since \((t,x)\) and \(Q\) were chosen arbitrarily, it follows that, for every \(\epsilon>0\),
    \(
    \widetilde{\mathcal{RA}}_{Q}^{\epsilon}\subseteq \widetilde{\mathcal{RA}}_{Q}.\)\\
\noindent \textbf{2)} Let \(\epsilon>0\) and \(z\in[0,T]\). At the terminal time \(T\),
\(
W_\epsilon(T,x,z)=\max\{\,c_1(T,x),\,c_2(T,x),\,-z,\,g(T,x)\,\},
\)
where \(c_1(T,\cdot),\,c_2(T,\cdot),\,g(T,\cdot)\) are signed–distance functions whose subzero level sets are the compact sets \(\mathbb{C}_{1,T},\,\mathbb{C}_{2,T},\,\mathbb{T}_T\); hence \(\{\,x: W_\epsilon(T,x,z)\le 0\,\}=\mathbb{C}_{1,T}\cap\mathbb{C}_{2,T}\cap\mathbb{T}_T\) is compact. Since \(W_\epsilon\) is Lipschitz in time, there exists \(L>0\) with \(\lvert W_\epsilon(T,x,z)-W_\epsilon(t,x,z)\rvert\le L\lvert T-t\rvert\) for all \(t\in[0,T]\). Because \(W_\epsilon(T,x,z)\to +\infty\) as \(\|x\|\to\infty\), we can choose \(R>0\) so that \(W_\epsilon(T,x,z)>LT\) whenever \(\|x\|>R\). It follows that for every \(t\in[0,T]\), \(\{\,x: W_\epsilon(t,x,z)\le 0\,\}\subseteq B_R(0);\) thus, for each fixed \(t\), the zero–sublevel set is bounded and, since \(W_\epsilon\) is continuous in \(x\), it is also closed, hence compact. Therefore, \(\widetilde{\mathcal{RA}}_Q\) with \(Q=z\) is contained in the compact set \([0,T]\times B_R(0)\). Define \(\mathcal{S}_\epsilon(Q):=\{(t,x)\in[0,T]\times B_R(0): -b< W(t,x,Q)\le 0\}\).
Therefore, \(
\widetilde{\mathcal{RA}}_{Q}^{\epsilon}\,\triangle\,\widetilde{\mathcal{RA}}_{Q}\subseteq \mathcal{S}_\epsilon(Q).\) This implies that
\(\!\lim_{\epsilon\downarrow 0}\,\mu\!\left(\!\widetilde{\mathcal{RA}}_{Q}^{\epsilon}\,\triangle\,\widetilde{\mathcal{RA}}_{Q}\!\right)
\!\le\!
\lim_{\epsilon\downarrow 0}\,\mu\!\left(\mathcal{S}_\epsilon(Q)\right).
\) Since \([0,T]\times B_R(0)\) has finite measure, we may apply the continuity from above property of the Lebesgue measure to the decreasing family
\(\{\mathcal{S}_\epsilon(Q)\}_{\epsilon>0}\):
for any sequence \(\!\epsilon_k\!\downarrow\! 0\),
\(
\lim_{k\to\infty}\mu\!\left(\mathcal{S}_{\epsilon_k}(Q)\right)
\!=\!\mu\!\left(\bigcap_{k=1}^\infty \mathcal{S}_{\epsilon_k}(Q)\right)
=\mu\!\left(\{(t,x)\in[0,T]\times B_R(0): W(t,x,Q)=0\}\right).
\) Therefore,
\[
\begin{aligned}
\lim_{\epsilon \downarrow 0}\,\mu\!\left(\widetilde{\mathcal{RA}}_{Q}^{\epsilon}\,\triangle\,\widetilde{\mathcal{RA}}_{Q}\right)
&\le \\
&\hspace{-5em}\mu\!\left(\{(t,x)\in[0,T]\times B_R(0):\, W(t,x,Q)=0\}\right).
\end{aligned}
\]
Hence, the limit vanishes whenever the zero-level set of \(W(\cdot,\cdot,Q)\) has measure zero. \hfill\(\blacksquare\)

\section*{Acknowledgment}
The authors thank Scott Moura, Jingqi Li, Marsalis Gibson, Jason Choi, and Bryce Ferguson for the insightful discussions.

\begin{IEEEbiography}[{\includegraphics[width=1in,height=1.25in,clip,keepaspectratio]{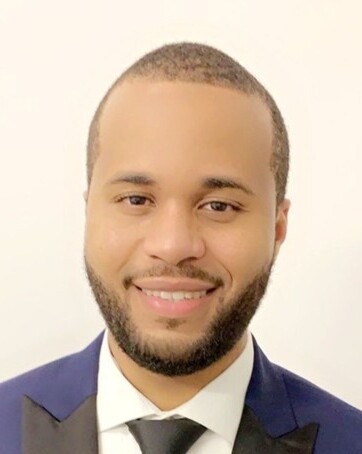}}]{Chams Eddine Mballo}
received the Ph.D.\ in aerospace engineering (2022) and an M.S.\ in mathematics (2021) from the Georgia Institute of Technology, Atlanta. He is currently a Postdoctoral Fellow at the University of California, Berkeley. His research focuses on safety‐critical control, reachability analysis, human–machine interaction, and sustainable aviation for next-generation electric aircraft.
\end{IEEEbiography}

\begin{IEEEbiography}[{\includegraphics[width=1in,height=1.25in,clip,keepaspectratio]{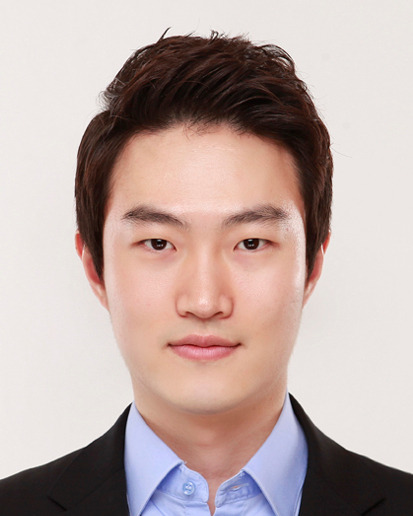}}]{Donggun Lee} received the Ph.D.\ degree in mechanical engineering from the University of California, Berkeley, CA, USA. He was a Postdoctoral Associate at the Massachusetts Institute of Technology (MIT), Cambridge, MA, USA. He is currently an Assistant Professor in Mechanical and Aerospace Engineering at North Carolina State University (NCSU), Raleigh, NC, USA. His research interests include control theory and machine learning, with applications to robotics and vehicles.
\end{IEEEbiography}

\begin{IEEEbiography}[{\includegraphics[width=1in,height=1.25in,clip,keepaspectratio]{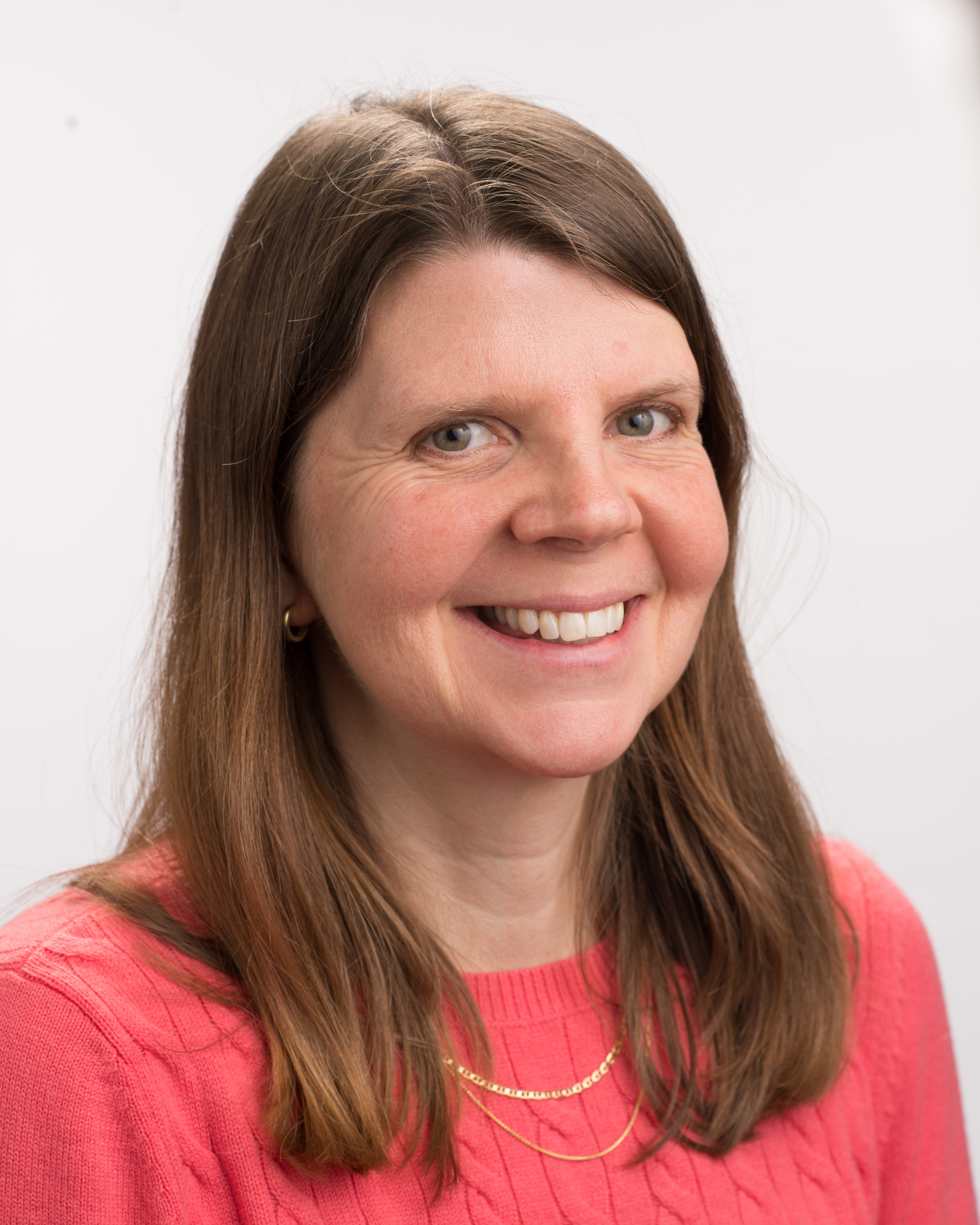}}]{Claire J. Tomlin}
(Fellow, IEEE) received the Ph.D.\ degree in electrical engineering and computer science from the University of California, Berkeley, CA, USA, in 1998. She is the James and Katherine Lau Professor of Engineering and the Professor and Chair with the Department of Electrical Engineering and Computer Sciences at UC Berkeley, Berkeley, CA, USA. From 1998 to 2007, she was an Assistant, Associate, and Full Professor in Aeronautics and Astronautics at Stanford University, Stanford, CA. In 2005, she joined UC Berkeley. Her research interests include control theory and hybrid systems, with applications to air traffic management, UAV systems, energy, robotics, and systems biology.\\
\indent Dr.~Tomlin is a MacArthur Foundation Fellow (2006). She was the recipient of the IEEE Transportation Technologies Award in 2017. In 2019, she was elected to the National Academy of Engineering and the American Academy of Arts and Sciences.
\end{IEEEbiography}
\end{document}